\documentclass[11pt]{article} 
\usepackage{times} 
\usepackage{amsmath,amsfonts}
\usepackage{epsfig,amssymb,amstext,xspace,theorem} 
\usepackage{graphicx,color}
\usepackage{verbatim}
\usepackage{multirow}

\newcommand{\np}{{\em NP}\xspace} 
\newcommand{\nphard}{\np-hard\xspace}
\newcommand{\npcomplete}{\np-complete\xspace} 
\newcommand{\apx}{{\em APX}\xspace}

\newtheorem{theorem}{Theorem}[section] 
\newtheorem{lemma}[theorem]{Lemma}

\newtheorem{claim}[theorem]{Claim}

\newtheorem{corollary}[theorem]{Corollary}

\newtheorem{fact}[theorem]{Fact}

{\theorembodyfont{\rmfamily} 

\newtheorem{algorithm}{Algorithm}
 
\newtheorem{definition}[theorem]{Definition}}

\def\blksquare{\rule{2mm}{2mm}} 
\def\qedsymbol{\blksquare}
\newcommand{\bg}[1]{\medskip\noindent{\bf #1}}
\newcommand{\ed}{{\hfill\qedsymbol}\medskip} 
\newenvironment{proof}{\bg{Proof : }}{\ed}

\newenvironment{proofof}[1]{\bg{Proof of #1 : }}{\ed}

\newenvironment{labellist}[1][A]
{\begin{list}{{#1}\arabic{enumi}.}{\usecounter{enumi}\addtolength{\leftmargin}{-1.5ex}}}
{\end{list}}

\newcommand{\Z}{\ensuremath{\mathbb Z}}

\newcommand{\C}{\ensuremath{\mathcal{C}}}

\newcommand{\Oc}{\ensuremath{\mathcal O}}
\newcommand{\Sc}{\ensuremath{\mathcal S}} 
\newcommand{\Pc}{\ensuremath{\mathcal P}}
\newcommand{\Qc}{\ensuremath{\mathcal Q}} 
\newcommand{\OPT}{\ensuremath{\mathit{OPT}}}
\newcommand{\OPTR}{\ensuremath{\mathit{OPT}_R}}
\newcommand{\lp}{\ensuremath{\text{LP}}}
\newcommand{\optrmax}{\ensuremath{R_{\lp}}}

\newcommand{\frall}{\ensuremath{\text{ for all }}} 

\newcommand{\sm}{\ensuremath{\setminus}} 
\newcommand{\es}{\ensuremath{\emptyset}}

\newcommand{\ceil}[1]{\ensuremath{\left\lceil#1\right\rceil}}
\newcommand{\floor}[1]{\ensuremath{\left\lfloor#1\right\rfloor}}
\newcommand{\poly}{\operatorname{\mathsf{poly}}}

\newcommand{\e}{\ensuremath{\epsilon}} 

\newcommand{\gm}{\ensuremath{\gamma}}

\newcommand{\sse}{\subseteq}

\newcommand{\mc}{\mathcal} 
 
\newcommand{\supp}{{\rm supp}}

\def\br#1{{{(#1)}}}

\newcommand{\nvrp}{\textsf{VRP}\xspace}
\newcommand{\vrp}{{\small\textsf{VRP}}\xspace}

\newcommand{\ndvrp}{\textsf{DVRP}\xspace}
\newcommand{\fdvrp}{{\footnotesize\textsf{DVRP}}\xspace}
\newcommand{\dvrp}{{\small\textsf{DVRP}}\xspace}

\newcommand{\nrvrp}{\textsf{RVRP}\xspace}
\newcommand{\frvrp}{{\footnotesize\textsf{RVRP}}\xspace}
\newcommand{\rvrp}{{\small\textsf{RVRP}}\xspace}
\newcommand{\mrvrp}{\ensuremath{\mathsf{RVRP}}\xspace}

\newcommand{\nkvrp}{$k$\textsf{RVRP}\xspace}

\newcommand{\kvrp}{$k${\small\textsf{RVRP}}\xspace}

\newcommand{\ktsp}{$k${\small\textsf{TSP}}\xspace}
\newcommand{\tsp}{{\small\textsf{TSP}}\xspace} 
\newcommand{\atsp}{{\small\textsf{ATSP}}\xspace} 
\newcommand{\matsp}{\ensuremath{\mathsf{ATSP}}}
\newcommand{\katspp}{$k${\small\textsf{ATSPP}}\xspace} 
\newcommand{\orient}{orienteering\xspace}
\newcommand{\morient}{\ensuremath{\mathsf{orient}}}
\newcommand{\mep}{{\small\textsf{MEP}}\xspace} 
\newcommand{\mmep}{\ensuremath{\mathsf{MEP}}}

\newcommand{\multi}{{multi-pair}\xspace}

\newcommand{\iopt}{\ensuremath{O^*}} 
\newcommand{\ioptr}{\ensuremath{R^*}} 
\newcommand{\into}{\ensuremath{\mathrm{in}}}

\newcommand{\comp}{\ensuremath{\Sc}}
\newcommand{\red}{\ensuremath{\mathsf{red}}}
\newcommand{\dist}{\ensuremath{D}}
\newcommand{\reg}{\ensuremath{\mathsf{reg}}}

\newcommand{\al}{\ensuremath{\alpha}}

\newcommand{\dt}{\ensuremath{\delta}}

\newcommand{\w}{\ensuremath{\omega}}

\newcommand{\tz}{\ensuremath{\tilde z}} 
\newcommand{\tP}{\ensuremath{\tilde P}} 
\newcommand{\hP}{\ensuremath{\hat P}}

\newcommand{\hz}{\ensuremath{\hat z}}

\newcommand{\Time}{\ensuremath{\mathsf{T}}}
\newcommand{\TS}{\ensuremath{\mathsf{TS}}}
\newcommand{\optdvlp}{\ensuremath{\OPT_{\text{\ref{dvrplp}}}}}

\title{Approximation Algorithms for Regret-Bounded Vehicle Routing and Applications to
Distance-Constrained Vehicle Routing}
\author{ 
    Zachary Friggstad\thanks{{\tt \{zfriggstad,cswamy\}@math.uwaterloo.ca}.  
    Dept. of Combinatorics and
    Optimization, Univ. Waterloo, Waterloo, ON N2L 3G1.  Supported in part by NSERC
    grant 327620-09.  The second author is also supported by an NSERC Discovery
    Accelerator Supplement Award and an Ontario Early Researcher Award.}  
\and
\addtocounter{footnote}{-1} 
    Chaitanya Swamy\footnotemark 
}
\date{}

\begin{document}

\maketitle 

\begin{abstract}
We consider vehicle-routing problems ({\footnotesize \nvrp}s) that incorporate the notion 
of {\em regret} of a client, which is a measure of the waiting time of a client relative
to its shortest-path distance from the depot.
Formally, we consider both the additive and multiplicative versions of, what we call, the 
{\em regret-bounded vehicle routing problem} (\frvrp).  
In these problems, we are given an undirected complete graph $G=(\{r\}\cup V,E)$ on $n$
nodes with a distinguished root (depot) node $r$, edge costs $\{c_{uv}\}$ that form a
metric, and a regret bound $R$. 
Given a path $P$ rooted at $r$ and a node $v\in P$, let $c_P(v)$ be the distance from $r$
to $v$ along $P$.
The goal is to find the fewest number of paths rooted at $r$ that cover all the nodes so
that for every node $v$ covered by (say) path $P$: 
(i) its additive regret $c_P(v)-c_{rv}$, with respect to $P$ is at most $R$ in 
{\em additive-\frvrp}; or  
(ii) its multiplicative regret, $c_P(c)/c_{rv}$, with respect to $P$ is at most $R$ in
{\em multiplicative-\frvrp}.

Our main result is the {\em first} constant-factor approximation algorithm for
additive-\frvrp. This is a substantial improvement over the previous-best 
$O(\log n)$-approximation.  
Additive-\frvrp turns out be a rather central vehicle-routing problem, whose 
study reveals insights into a variety of other regret-related problems as well as the 
classical {\em distance-constrained {\footnotesize \nvrp}} ({\fdvrp}),
enabling us to obtain guarantees for these various problems by leveraging our algorithm 
for additive-\frvrp and the underlying techniques. 
We obtain approximation ratios of $O\bigl(\log(\frac{R}{R-1})\bigr)$ for 
multiplicative-\frvrp, and 
$O\bigl(\min\bigl\{\OPT,\frac{\log D}{\log\log D}\bigr\}\bigr)$ for 
\fdvrp with distance bound $D$ via reductions to additive-\frvrp; the
latter improves upon the previous-best approximation for \fdvrp.

A noteworthy aspect of our results is that they are obtained by devising rounding
techniques for a natural {\em configuration-style LP}. This furthers 
our understanding of LP-relaxations for {\footnotesize \nvrp{}}s and enriches the toolkit 
of techniques that have been utilized for configuration LPs.   
\end{abstract}

\section{Introduction} \label{intro}
Vehicle-routing problems (\vrp{}s) constitute a broad class of combinatorial-optimization
problems that find a wide range of applications and have been widely studied in the
Operations Research and Computer Science communities (see, 
e.g.,~\cite{HaimovichK85,LiLD92,TothV02,CharikarR98,BlumCKLMM07,BansalBCM04,NagarajanR12,ChekuriKP12}
and the references therein). These problems are
typically described as follows. There are one or more vehicles 
that start at some depot and provide service to an underlying set of of clients, and the
goal is to design routes for the vehicles that visit the clients as quickly as possible.
The most common way of formalizing the objective of minimizing client delays is to seek
a route of minimum length, or equivalently, a route that minimizes the maximum client
delay, which gives rise to (the path variant of) the celebrated traveling salesman problem
(\tsp). However, this objective does not differentiate between clients located at
different distances from the depot, and 
a client closer to the depot
may end up incurring a larger delay than a client that is further away,
which can be considered a source of unfairness and hence, client dissatisfaction.
Adopting a client-centric approach, we consider an alternate objective that addresses
this unfairness and seeks to design routes that promote customer satisfaction. 

Noting that the delay of a client is inevitably at least the shortest-path
distance from the depot to the client location, following~\cite{SpadaBL05,ParkK10}, we
seek to ensure that the {\em regret} of a client, which is a measure of its waiting time
{\em relative to its shortest-path distance from the depot}, is bounded.  
More precisely, we consider the following genre of vehicle-routing problems. We
are given an undirected complete graph $G=(\{r\}\cup V,E)$ on $n$ nodes with a
distinguished root (depot) node $r$, and metric edge costs or distances $\{c_{uv}\}$. 
Given a path $P$ rooted at $r$ and a node $v\in P$, let $c_P(v)$ be the distance from $r$
to $v$ along $P$ (i.e., the length of the $r$-$v$ subpath of $P$).
There are two natural ways of comparing $c_P(v)$ and $c_{rv}$ to define the regret of a
node $v$ on path $P$. We define the {\em additive regret} of $v$ with respect to $P$ to be
$c_P(v)-c_{rv}$,%
\footnote{The distinction between the delay and additive regret of a client is akin to the
distinction between the completion time and flow time of a job in scheduling problems.} 
and the {\em multiplicative regret} of $v$ with respect to $P$ to be $c_P(v)/c_{rv}$.
We are also given a regret bound $R$. Fixing a regret measure, 
a feasible solution is a collection of paths rooted at
$r$ that cover all the nodes in $G$ such that the regret of every node with
respect to the path covering it is at most $R$.  
Thus, a feasible solution to: (i) the additive-regret problem yields the satisfaction
guarantee that every client $v$ is visited by time $c_{rv}+R$; and 
(ii) the multiplicative-regret problem ensures that every client $v$ is visited by time
$c_{rv}\cdot R$. 
The goal is to find a feasible solution that uses the fewest number of paths.
We refer to these two problems 
as {\em additive-regret-bounded \vrp} (additive-\rvrp) and 
{\em multiplicative-regret-bounded \vrp} (multiplicative-\rvrp) respectively.

Additive-\rvrp has been sometimes referred to as the {\em schoolbus problem} in the 
literature~\cite{SpadaBL05,ParkK10,BockGKS11}. However, this term is used to refer to an
umbrella of vehicle-routing problems, some of which do not involve regret, so we use the
more descriptive name of additive-\rvrp. 
Both versions of \rvrp 
are \apx-hard via simple reductions from \tsp and \tsp-path (Theorem~\ref{symkhard}), so
we focus on approximation algorithms.    

\paragraph{Our results.}
We undertake a systematic study of regret-related vehicle-routing problems 
from the perspective of approximation algorithms. 
As we illustrate below, additive-\rvrp turns out to be the more fundamental of the above 
two problems and a rather useful problem to investigate, and our study 
yields insights and techniques that can be applied, often in a black-box fashion, to
derive algorithms for various vehicle-routing problems, which include both regret-related
problems, and classical problems such as distance-constrained vehicle routing. 
We therefore focus on additive regret;
unless otherwise stated, regret refers to additive regret, and a regret-related problem
refers to the problem under the additive-regret measure.  

Our main result is the {\em first constant-factor} approximation algorithm for (additive) 
\rvrp (Theorem~\ref{minkapx}). 
This is a substantial improvement over the previous-best $O(\log n)$-approximation ratio 
for \rvrp obtained in~\cite{BockGKS11} via the standard set-cover greedy algorithm and
analysis. 

\begin{table}[t!]
\begin{center}
\small
\hspace*{-4ex}
\begin{tabular}{|ccccc||cc|} \hline
\multicolumn{5}{|c||}{Symmetric metrics} & \multicolumn{2}{c|}{Asymmetric metrics} 
\\ \hline
{\bf\footnotesize \nrvrp} & {\footnotesize \nkvrp} 
& Multiplicative-{\footnotesize\nrvrp} & Multiplicative-{\footnotesize\nkvrp} 
& {\footnotesize \ndvrp} 
& {\footnotesize \nrvrp} & {\footnotesize \nkvrp} 
\\ 
$\mathbf{31}$ 
& $O(k^2)$ & $O\bigl(\log(\frac{R}{R-1})\bigr)$ & $O(1)$  
& $O\bigl(\min\bigl\{\OPT,\frac{\log D}{\log\log D}\bigr\}\bigr)$ 
& $O(\log n)$ & $O(k^2\log n)$ 
\\ \hline
\end{tabular}
\caption{\small Summary of our results. Our main result, for {\footnotesize \nrvrp},
yields guarantees for other symmetric-metric problems.} 
\end{center}
\vspace{-3ex}
\end{table}

A noteworthy aspect of our result is that 
we develop linear-programming (LP) based techniques for the problem.
While LP-relaxations have been exploited with striking
success in the design and analysis of approximation algorithms, our understanding of
LP-relaxations for \vrp{}s is quite limited 
(with \tsp, and the minimum-latency problem to a lesser extent, being the exceptions), 
and this has been a stumbling block in the design of approximation algorithms for many of
these problems.   
Notably, we develop LP-rounding techniques for 
a natural {\em configuration-style LP-relaxation} for \rvrp, 
which is an example of the set-partitioning model for vehicle routing with
time windows (see~\cite{TothV02}). While it is not difficult to come up with such
(approximately-solvable) configuration LPs for vehicle-routing problems, and they have
been observed computationally to provide excellent lower bounds on the optimal
value~\cite{DesrochersDS92}, there are few theoretical bounds on the effectiveness of
these LPs. Moreover, the limited known guarantees (for general metrics) typically only
establish logarithmic bounds~\cite{NagarajanR08,BockGKS11}, which follow from the
observation that the configuration LP can be viewed as a standard set-cover LP.   
In contrast, we exploit the structure of our configuration LP for \rvrp using
novel methods and prove a {\em constant} integrality gap for the relaxation, which
serves to better justify the good empirical performance of these LPs. 
Although configuration LPs are often believed to be powerful,  
they have been leveraged only sporadically in the design of approximation algorithms; some 
notable exceptions
are~\cite{KarmarkarK82,BansalS06,Svensson12,Rothvoss13,DobzinskiNS10,Feige09}.
Our work contributes to the toolkit of techniques that have been utilized for
configuration LPs, and our techniques may find wider applicability. 

We use our algorithm for additive-\rvrp 
to obtain an $O\bigl(\log(\frac{R}{R-1})\bigr)$-approximation for multiplicative-\rvrp  
with regret-bound $R$ (Theorem~\ref{multrvrp}). 
\nolinebreak
\mbox{Thus, we obtain a constant-factor approximation 
for any fixed $R$.} 

Interestingly, our algorithm for \rvrp yields {\em improved guarantees} for
(the path-variant of) the classical {\em distance-constrained  
vehicle-routing problem}
(\dvrp)~\cite{LaporteDN84,LiLD92,NagarajanR08,NagarajanR12}---find the fewest   
number of rooted paths of {\em length} at most $D$ that cover all the nodes---%
via a reduction to \rvrp. 
({\dvrp} usually refers to the version where we seek tours containing the
root; \cite{NagarajanR08} shows that the path- and tour-versions are within a factor of 2
in terms of approximability.) 
We obtain an $O\bigl(\frac{\log R_{\max}}{\log\log R_{\max}}\bigr)$-approximation for
\dvrp (Theorem~\ref{dvrpthm}), where $R_{\max}\leq D$ is the maximum regret of a node in
an optimal solution, 
which improves upon the previous-best $O(\log D)$-guarantee for \dvrp~\cite{NagarajanR08}.  
We believe that this reduction is of independent interest.
Exploiting our LP-based guarantee for \rvrp, we obtain the same integrality-gap bound for
the natural configuration LP for \dvrp. 
We also show that the integrality gap of the configuration LP is $O(\OPT_\lp)$, where
$\OPT_\lp$ is the optimal value of the LP. This is interesting because for the standard
set-cover LP, there are $O(\log n)$-integrality-gap examples even when the optimal
LP-value is a constant; although the \dvrp-LP is also a  set-cover LP, our result
precludes such an integrality-gap construction for this LP  
and raises the enticing possibility that the additional structure in \dvrp can be
further exploited, perhaps by refining our methods,
to derive improved guarantees.

We leverage our techniques to obtain guarantees for various variants and generalizations
of \rvrp (Section~\ref{extn}), including, most notably, 
(i) the variants where we fix the number $k$ of rooted paths (used to cover the nodes) and
seek to minimize the maximum additive/multiplicative regret of a node, which we refer to
as {\em additive/multiplicative- \kvrp}; 
and 
(ii) (additive) \rvrp and \kvrp in {\em asymmetric metrics}. 

We obtain an $O(k^2)$-approximation for additive-\kvrp 
(Theorem~\ref{symregapx}), which is the {\em first} approximation guarantee for \kvrp. 
Previously, the only approximation results known for \kvrp were for the special
cases where we have a tree metric~\cite{BockGKS11} (note that the $O(\log n)$-distortion
embedding of general metrics into tree metrics does not approximate regret), and when 
$k=1$~\cite{BlumCKLMM07}. 
In particular, {\em no} approximation guarantees were known previously 
{\em even when $k=2$}; in contrast, we achieve a constant-factor approximation for any
fixed $k$. Partially complementing this result, we show that the integrality gap of the
configuration LP for \kvrp is $\Omega(k)$ (Theorem~\ref{intgap}).
Multiplicative-\kvrp turns out to be an easier problem, and the LP-rounding
ideas in~\cite{ChakrabartyS11} yield an $O(1)$-approximation for this problem
(Theorem~\ref{multregapx}).  

For asymmetric metrics, we exploit the simple but key observation that regret can be
captured via a suitable  asymmetric metric that we call the {\em regret metric} $c^{\reg}$
(see Fact~\ref{obs:regret}). This alternative view of regret yields surprising dividends,
since we can directly plug in results for asymmetric metrics to obtain results for regret
problems. In particular, results for $k$-person asymmetric $s$-$t$
TSP-path~\cite{FriggstadSS10,Friggstad11} translate to results for asymmetric \rvrp and
\kvrp, and we achieve approximation ratios of $O(\log n)$ and $O(k^2\log n)$ respectively
for these two problems. 
Although regret metrics form a strict subclass of asymmetric metrics, 
we uncover an interesting connection between the approximability
of asymmetric \rvrp and \atsp. We show that an $\al$-approximation for asymmetric \rvrp
implies a $2\al$-approximation for \atsp (Theorem~\ref{asymkhard});  
thus an $\w(\log\log n)$-improvement to the approximation we achieve for asymmetric
{\rvrp would improve the current best 
$O\bigl(\frac{\log n}{\log\log n}\bigr)$-approximation for \atsp~\cite{AsadpourGMOS10}.}

\paragraph{Our techniques.}
Our algorithm for additive-\rvrp (see Section~\ref{sec:symmetric}) is based on rounding a
fractional solution to a natural configuration LP \eqref{minklp}, where we have a variable
for every path of regret at most $R$ and we enforce that every node is covered to an
extent of 1 by such paths. 
Although this LP has an exponential number of variables, we can obtain a near-optimal
solution $x^*$ by using an approximation algorithm for 
{\em orienteering}~\cite{BlumCKLMM07,ChekuriKP12} (see ``Related work'') 
to provide an approximate separation oracle for the dual LP. 

Let $k^*=\sum_P x^*_P$.
To round $x^*$, we first observe that it suffices
to obtain $O(k^*)$ paths of {\em total regret} $O(k^*R)$ (see Lemma~\ref{avg2max}).  
At a high level, we would ideally like to ensure that directing the paths in the 
support of $x^*$ away from the root yields a directed acyclic graph $H$. If we have
this, then by viewing $x^*$ as the path decomposition of a flow in $H$, and by the
integrality property of flows, we can round $x^*$ to an integral flow that covers all the
nodes, has value at most $\ceil{k^*}$, and whose cost in the regret metric is at most the
$c^{\reg}$-cost of $x^*$, which is at most $k^*R$. This integral flow decomposes into a
collection of $\ceil{k^*}$ paths that cover $V$ (since $H$ is acyclic), which yields the  
desired rounding. 

Of course, in general, we will not be in this ideal situation. Our goal will be to
identify a subset $W$ of ``witness nodes'' such that:
(a) $x^*$ can be converted into a fractional solution that covers $W$ and has the above
acyclicity-property without blowing up the $c^{\reg}$-cost by much; and    
(b) nodes in $V\sm W$ can be attached to $W$ incurring only an $O(k^*R)$ cost.   
The new fractional solution can then be rounded to obtain integral paths that cover $W$,
which in turn can be extended so that they cover $V$. 
In achieving this goal, we gain significant leverage from the fact that the
configuration LP yields a collection of fractional {\em simple} paths that cover all the
nodes, which is a stronger property than having a flow where every node has at least one
unit of incoming flow. 
We build a forest $F$ of cost $O(k^*R)$ and select one node from each component of $F$ as
a witness node; this immediately satisfies (b). The construction ensures that:
first, every witness node $w$ has an associated collection of ``witness paths'' that cover
it to a large extent, say, $\frac{1}{2}$; and second, for every path $P$, the witness nodes
that use $P$ as a witness path have strictly increasing distances from the root $r$ and
occur on $P$ in order of their distance from $r$. It follows that by shortcutting each
path to only contain the witness nodes that use the path as a witness path, and blowing up
the $x^*$ values by 2, we achieve property (a). 

\medskip

Our algorithms for multiplicative-\rvrp and \dvrp 
capitalize on the following insight. 
Suppose there exist $k$ paths covering a given set $S$ of nodes and incurring
additive regret at most $\rho$ for these nodes. Then, for any $\e>0$, one can use our
algorithm for \rvrp to find $O\bigl(\frac{k}{\e}\bigr)$ paths covering $S$ such that the
nodes in $S$ have additive regret at most $\e\rho$ (Lemma~\ref{covering}). For
multiplicative-\rvrp with regret bound $R$, 
we apply this observation to every ``ring'' $V_i:=\{v: c_{rv}\in[2^{i-1},2^i)\}$ to obtain
$O(\OPT)$ paths covering $V_i$ such that the $V_i$-nodes face at most $(R-1)\cdot 2^{i-2}$
additive regret. This follows since the optimal solution
covers $V_i$ inducing additive regret at most $(R-1)\cdot 2^i$ for these nodes.
Concatenating the paths obtained for the $V_i$s whose indices are $O\bigl(\log(\frac{R}{R-1})\bigr)$
apart yields the $O\bigl(\log(\frac{R}{R-1})\bigr)$-approximation.  

For \dvrp, we build upon the above insight. Rather than fixing beforehand, as above, the
regret bounds and the corresponding node-sets to cover via paths ensuring that regret
bound, we use a dynamic-programming approach.
Crucially, in the analysis, we bound the number of paths needed to cover a set of nodes
with a given regret bound by suitably {\em modifying} the paths of a 
{\em structured near-optimal solution} $\Oc$.  
We argue that a specific choice (depending on $\Oc$) of regret bounds and node-sets
yields an $O\bigl(\frac{\log R_{\max}}{\log\log R_{\max}}\bigr)$-approximation. 
In doing so, we argue that each choice of regret-bound is such that we make progress by
decreasing substantially either the regret-bound {\em or} number of paths needed.
Since our \rvrp-algorithm is in fact LP-based, this also yields a bound on the integrality
gap of the natural configuration LP for \dvrp.

For the $O(\OPT_\lp)$ integrality-gap result for \dvrp, 
we show that one can partition the nodes so 
that for each part $S$,
there is a distinct node $t_S$ such that the paths ending at $t_S$
cover the $S$-nodes to an extent of $\Omega\bigl(\frac{1}{\OPT_\lp}\bigr)$.
Multiplying the LP-solution by $O(\OPT_\lp)$ then yields a fractional solution that
covers the $S$-nodes  
{incurring regret at most $D-c_{rt_S}$, which we can round using our \rvrp-algorithm.}  

\paragraph{Related work.}
There is a wealth of literature on vehicle routing problems (see, e.g.,~\cite{TothV02}),
and the survey~\cite{ParkK10} discusses a variety of problems under the umbrella of
schoolbus-routing problems; 
we limit ourselves to the work that is relevant to our problems. 
The use of regret as a vehicle-routing objective seems to have been first considered
in~\cite{SpadaBL05}, who present various heuristics and empirical results.

Bock et al.~\cite{BockGKS11} developed the first approximation algorithms for \rvrp, 
but focus mainly on tree metrics, for which they achieve a 3-approximation. For general
metrics, they observe that \rvrp can be cast as a covering problem,
and finding a minimum-density set is an {\em orienteering}
problem~\cite{GoldenLV87,BlumCKLMM07}:  
given node rewards, end points $s$, $t$, and a length bound $B$, find an $s$-$t$
path of length at most $B$ that gathers maximum total node-reward. 
Thus, the greedy set-cover algorithm combined with a suitable $O(1)$-approximation for 
orienteering~\cite{BlumCKLMM07,ChekuriKP12} immediately yields an $O(\ln n)$-approximation
for \rvrp. Previously, this was the best approximation algorithm for \rvrp in general
metrics.  
For \kvrp, {\em no} previous results were known for general metrics, even when $k=2$.  
(Note that we obtain a constant approximation for \kvrp for any fixed $k$.) 
\cite{BockGKS11} obtain a 12.5-approximation for \kvrp in tree metrics. 
When $k=1$, \kvrp becomes as a special case of the {\em min-excess path} problem
\nolinebreak
\mbox{introduced by~\cite{BlumCKLMM07}, who devised a $(2+\e)$-approximation for this
problem.}

To the best of our knowledge, multiplicative regret, and the asymmetric
versions of \rvrp and \kvrp have not been considered previously. 
Our algorithm for multiplicative-\kvrp uses the LP-based techniques developed
by~\cite{ChakrabartyS11} for the minimum latency problem.  
The set-cover greedy algorithm can also be applied to asymmetric 
\rvrp. This yields approximation ratios of $O\bigl(\frac{\log^3 n}{\log\log n}\bigr)$ in
polytime, and $O(\log^2 n)$ in {\em quasi-polytime} using the 
$O\bigl(\frac{\log^2 n}{\log\log n}\bigr)$- and $O(\log\OPT)$- approximation algorithms
for directed orienteering in~\cite{NagarajanR11} and~\cite{ChekuriP05} respectively.
Both factors are significantly worse than the $O(\log n)$-approximation that we obtain via
an easy reduction to \katspp (find $k$ $s$-$t$ paths of minimum total cost that cover all 
nodes).  
Friggstad et al.~\cite{FriggstadSS10} obtained the first results for \katspp 
which were later improved by~\cite{Friggstad11} to an $O(k \log n)$-approximation
\nolinebreak
\mbox{and a bicriteria result that achieves $O(\log n)$-approximation using at most $2k$
paths.} 

Replacing the notion of client-regret in our problems with client-delay 
gives rise to some well-known vehicle-routing and TSP problems. The client-delay version
of \rvrp corresponds to (path-) \dvrp. 
Nagarajan and Ravi~\cite{NagarajanR08} give 
an $O(\log\min\{D,n\})$-approximation for general metrics, 
and a 2-approximation for trees.
Obtaining a constant-factor approximation for \dvrp in general metrics has been a long-standing
open problem. 
As noted earlier, regret can be captured by the asymmetric regret metric and
thus \rvrp is {\em precisely} (path-) \dvrp in the regret metric.    
Thus, our work yields an $O(1)$-approximation for \dvrp in this specific
asymmetric metric.
We find this to be quite interesting and surprising since one would normally expect that 
\dvrp would become {\em harder} in an asymmetric metric. 

The client-delay version of \kvrp 
yields the \ktsp problem of finding $k$ rooted paths of minimum maximum cost that
cover all nodes, which admits a constant-factor approximation via a reduction to \tsp.  

The orienteering problem plays a key role in vehicle-routing problems, including our
algorithm for \rvrp where it yields an approximate separation oracle for the
dual LP. Blum et al.~\cite{BlumCKLMM07} obtained the first constant-factor approximation
algorithm for orienteering, 
and the current best approximation is $2+\e$ due to Chekuri et al.~\cite{ChekuriKP12}.  
\cite{NagarajanR11,ChekuriKP12} study (among other problems) directed orienteering 
and obtain approximation ratios of $O\bigl(\frac{\log^2 n}{\log\log n}\bigr)$ and
$O(\log^2\OPT)$ respectively. The backbone of all of these algorithms is the min-regret
$K$-path problem (called the min-excess path problem in~\cite{BlumCKLMM07})---choose a
min-regret path covering at least $K$ nodes---which captures \kvrp when $k=1$.%
\footnote{Viewed from the perspective of the regret metric, the min-regret $K$-path 
problem {\em trivially reduces} to the min-cost $K$-path problem (choose a min-cost path
covering at least $K$ nodes) in asymmetric metrics. This allows one to slightly improve
Theorem 8 in~\cite{NagarajanR11} and Lemma 2.4 in~\cite{ChekuriKP12}.}
\cite{ChekuriP05} used a different approach and gave a quasi-polytime
$O(\log\OPT)$-approximation for directed orienteering. 
Finally, Bansal et al.~\cite{BansalBCM04} and Chekuri et al.~\cite{ChekuriKP12} consider
orienteering with time windows, 
where nodes have time windows and we seek to maximize the number of nodes that are visited
in their time windows, and its special case where nodes have deadlines, both of which
generalize orienteering. They obtain polylogarithmic approximation ratios for these problems.

\section{Preliminaries} \label{prelim}
Recall that an instance of \rvrp is specified by a complete undirected graph 
$G=(\{r\}\cup V,E)$, where $r$ is a distinguished root node, with metric edge costs
$\{c_{uv}\}$, and a regret-bound $R$. Let $n=|V|+1$. 
We call a path in $G$ rooted if it begins at $r$. 
Unless otherwise stated, we think of the nodes on $P$ as being ordered in
increasing order of their distance along $P$ from $r$, 
and directing $P$ away from $r$ means that we direct each edge $(u,v)\in P$ from $u$ to
$v$ if $u$ precedes $v$ (under this ordering). 
We use $\dist_v$ to denote $c_{rv}$ for all $v\in V\cup\{r\}$. 
For a set $S$ of edges, we sometimes use $c(S)$ to denote $\sum_{e\in S}c_e$. 
By scaling, and merging all nodes at distance 0 from each other, we may assume that
$c_{uv}$ is a positive integer for all $u,v\in V\cup\{r\}$.
Thus, $\dist_v\geq 1$ for all $v\in V$. 
Unless otherwise qualified, regret refers to additive regret in the sequel.

It will be convenient to assume
that $R>0$: if $R=0$ then we can determine 
whether an edge $(u,v)$ lies on a shortest rooted path, and if so direct $(u,v)$ as
$u\rightarrow v$ if $\dist_v=\dist_u+c_{uv}$, to obtain a directed acyclic graph (DAG)
$H$. Our problem then reduces to finding the minimum number of directed rooted paths in
$H$ to cover all the nodes, which can be solved efficiently using network-flow techniques.  

The following equivalent way of viewing regret will be convenient. For every ordered pair
of nodes $u,v \in V \cup \{r\}$, define the {\em regret distance} (with respect to $r$) to
be $c^{\reg}_{uv}:=\dist_u+c_{uv}-\dist_v$. 

\begin{fact} \label{obs:regret}
(i) The regret distances $c^{\reg}_{uv}$ are nonnegative and satisfy the triangle inequality:
$c^{\reg}_{uv} \leq c^{\reg}_{uw}+c^{\reg}_{wv}$ for all $u,v,w \in V \cup \{r\}$. Hence,
$\{c^{\reg}_{uv}\}$ forms an asymmetric metric that we call the {\em regret metric}.

\noindent
(ii) For a $u\leadsto v$ path $P$, we have 
$c^{\reg}(P):=\sum_{e\in P}c^{\reg}_e=\dist_u+c(P)-\dist_v$, and for a cycle $Z$, we have
$c^{\reg}(Z)=c(Z)$.
Properties (i) and (ii) hold even when the underlying $\{c_{uv}\}$ metric is asymmetric.
\end{fact}

We infer from Fact~\ref{obs:regret} that if $P$ is a rooted path and $v\in P$, then the
regret of $v$ with respect to $P$ is simply the $c^{\reg}$-distance to $v$ along $P$,
which we denote by $c^{\reg}_P(v)$, and the regret of nodes on $P$ cannot decrease as one 
moves away from the root (since $c^{\reg}\geq 0$). 
We define the regret of $P$ to be the regret of the end-node of $P$, 
which by part (ii) of Fact~\ref{obs:regret} is given by $c^{\reg}(P)=\sum_{e\in P}c^{\reg}_e$. 

Lemma~\ref{avg2max} makes the key observation that one can always convert a collection of
paths with {\em average regret} at most $\al R$ into one where every path has regret at
most $R$ by blowing up the number of paths by an $(\al+1)$ factor, and hence, it
suffices to obtain a near-optimal solution with average regret $O(R)$.   

\begin{lemma} \label{cor:average} \label{avg2max}
Given rooted paths $P_1, \ldots, P_k$ with total regret $\alpha k R$,
we can efficiently find at most $(\alpha+1) \cdot k$ rooted paths, each 
regret at most $R$, that cover $\bigcup_{i=1}^k P_i$.
\end{lemma}

\begin{proof}
Let $\al_1 R,\ldots,\al_k R$ be the regrets of $P_1,\ldots,P_k$ respectively.
We show that for each path $P_i$, we can obtain $\max\{\ceil{\al_i},1\}$ rooted paths of
regret at most $R$ that cover the nodes of $P_i$. 
Applying this to each path $P_i$, we obtain
at most $\sum_{i=1}^k(\alpha_i+1)=(\alpha+1)\cdot k$ rooted
paths with regret at most $R$ that cover $\bigcup_{i=1}^k P_i$.

Fix a path $P_i$. If $\al_i\leq 1$, there is nothing to be done, so assume
otherwise. The idea is to simply break $P_i$ at each point where the regret exceeds a
multiple of $R$, and connect the starting point of each such segment directly to $r$. 
More formally, 
for $\ell=1,\ldots,\beta_i:=\ceil{\al_i}-1$, let $v_\ell$ be the first node on $P$
with $c^{\reg}_P(v)>\ell R$, and let $u_{\ell-1}$ be its (immediate) predecessor on $P$.  
Let $v_0=r$ and $u_{\beta_i}$ be the end point of $P_i$.
We create the $\ceil{\al_i}$ paths given by $r,v_\ell\leadsto u_\ell$ for
$\ell=0,\ldots,\beta_i$, which clearly together cover the nodes of $P_i$. 
The regret of each such path is
$c^{\reg}_{rv_\ell}+c^{\reg}_P(u_\ell)-c^{\reg}_P(v_\ell)=c^{\reg}_P(u_\ell)-c^{\reg}_P(v_\ell)\leq
(\ell+1)R-\ell R=R$, where the last inequality follows from the definitions of $v_\ell$,
$v_{\ell+1}$ and $u_\ell$ (which precedes $v_{\ell+1}$).
\end{proof}

Approximation algorithms for symmetric \tsp variants often exploit the fact that edges
may be traversed in any direction, to convert a connected subgraph 
into an Eulerian tour while losing a factor of 2 in the cost.
This does not work for \rvrp since $c^{\reg}$ is an asymmetric metric.
Instead, we exploit a key observation of Blum et al.~\cite{BlumCKLMM07}, who identify
portions of a rooted path $P$ whose total $c$-cost can be charged to $c^{\reg}(P)$. 

\begin{definition} \label{redblue}
Let $P$ be a rooted path ending at $w$. Consider an edge $(u,v)$ of $P$, where $u$ 
precedes $v$ on $P$. We call this a {\em red} edge of $P$ if there exist nodes $x$ and $y$  
on the $r$-$u$ portion and $v$-$w$ portion of $P$ respectively such that
$\dist_x\geq\dist_y$; otherwise, we call this a {\em blue} edge of $P$. 
For a node $x\in P$, let $\red(x,P)$ denote the maximal subpath $Q$ of $P$ containing $x$ 
consisting of only red edges (which might be the trivial path $\{x\}$).
\end{definition}

We call a maximal blue/red subpath of a rooted path $P$ a blue/red interval of $P$.
The blue and red intervals of $P$ correspond roughly 
to the type-1 and type-2 segments of $P$, as defined in~\cite{BlumCKLMM07}. Distinguishing
the edges on $P$ as red or blue serves two main purposes. 
First, the total cost of the red edges is proportional to the regret of $P$
(Lemma~\ref{lem:redcost}). Second, if we shortcut $P$ so that it contains only one 
node from each red interval, then the resulting edges must all be distance
increasing (Lemma~\ref{dinc}). Consequently, if we perform this operation on a collection 
of paths and direct edges away from the root, then we obtain a DAG. 

\begin{lemma}[Blum et al.~\cite{BlumCKLMM07}]
\label{lem:redcost} 
For any rooted path $P$, we have
$\sum_{\text{$e$ red on $P$}} c_e \leq \frac{3}{2} c^{\reg}(P)$.
\end{lemma}

\begin{proof} 
Each red edge is contained in a ``type-2 segment'', as defined in~\cite{BlumCKLMM07},
and Corollary 3.2 in~\cite{BlumCKLMM07} proves that the total length of type-2 segments is
at most $\frac{3}{2}\cdot c^{\reg}(P)$.  
\end{proof}

\begin{lemma} \label{lem:cross} \label{dinc}
(i) Suppose $u, v$ are nodes on a rooted path $P$ such that $u$ precedes $v$ on $P$ and
$\red(u,P) \neq \red(v,P)$, then $\dist_{u} < \dist_{v}$.
(ii) Hence, if $P'$ is obtained by shortcutting $P$ so that it contains at most one node from
each red interval of $P$, then for every edge $(x,y)$ of $P'$ with $x$ preceding $y$ on
$P'$, we have $\dist_x < \dist_y$.
\end{lemma}

\begin{proof}
Since $u$ precedes $v$ on $P$ and $\red(u,P)\neq \red(v,P)$, there must be some edge 
$(a,b)\in P$ such that $(a,b)$ is blue on $P$, and $a$, $b$ lie on the $u$-$v$ portion of
$P$ (note that it could be that $a=u$ and/or $b=v$). So if $\dist_u\geq\dist_v$ then 
$(a,b)$ would be classified as red.
Part (ii) follows immediately from part (i).
\end{proof}

\paragraph{Orienteering.}
Our algorithms are based on rounding the solution to an exponential-size LP-relaxation of
the problem. A near-optimal solution to this LP can be obtained by solving the dual LP
approximately. The separation oracle for the dual LP corresponds to a 
{\em point-to-point orienteering} problem, which is defined as follows. We are
given an undirected complete graph with nonnegative node-rewards, edge lengths that form a
metric, origin and destination nodes $s$, $t$, and a length bound $B$. The goal is to find
an $s$-$t$ path $P$ of total length at most $B$ that gathers maximum total reward. 
In the {\em rooted orienteering} problem, we only specify the origin $s$, and a path
rooted at $s$. Unless otherwise stated, we use orienteering to mean point-to-point
orienteering. Clearly, an algorithm for \orient can also be used for rooted
orienteering. 
A related problem is the {\em min-excess path} (\mep) problem defined
by~\cite{BlumCKLMM07}, where we are given $s$, $t$, and a target reward $\Pi$, and we seek
to find an $s$-$t$ path of minimum regret that gathers reward at least $\Pi$. 

In the unweighted version of these problems, all node rewards are 0 or 1. Observe that
the weighted versions of these problems can be reduced to their unweighted version in
pseudopolynomial time by making co-located copies of a node. For \orient, by suitably
scaling and rounding the node-rewards, one can obtain a 
$\poly\bigl(\text{input size},\frac{1}{\e}\bigr)$-time reduction where we lose a 
$(1+\e)$-factor in approximation. For \mep, this data rounding yields a bicriteria
approximation where we obtain an $s$-$t$ path with reward at least $\Pi/(1+\e)$. 
Both the unweighted and weighted versions of \orient and \mep are \nphard. The current
best approximation factors for these problems are $(2+\e)$ for \orient due to Chekuri et 
al.~\cite{ChekuriKP12}, and $(2+\e)$ for {\em unweighted} \mep due to Blum et
al.~\cite{BlumCKLMM07}, for any positive constant $\e$.

\section{An LP-rounding constant-factor approximation for (additive) {\large \nrvrp}} 
\label{sec:symmetric} 
We consider the following configuration-style LP-relaxation for \rvrp, which was also
mentioned in~\cite{BockGKS11}.
Let $\C_R$ denote the collection of all rooted paths with regret at most $R$.  
We introduce a variable $x_P$ for each path $P \in \C_R$ to denote if path $P$ is chosen.  
Throughout, we use $P$ to index paths in $\C_R$. 
\begin{equation}
\min \quad \sum_{P} x_P \qquad 
\text{s.t.} \qquad \sum_{P: v \in P} x_P \geq 1 \quad \frall v\in V, \qquad 
x_P \geq 0 \quad \frall P. \tag{P} \label{minklp}
\end{equation}
Let $\OPT$ denote the optimal value of \eqref{minklp}. 
Note that $\OPT\geq 1$. 
It is easy to give a reduction from \tsp showing that it is \npcomplete to decide if there 
is a feasible solution that uses only 1 path; hence, it is \nphard to achieve an
approximation factor better than 2 (Theorem~\ref{symkhard}). 
Complementing this, we devise an algorithm for \rvrp based on LP-rounding 
that achieves a constant approximation ratio (and thus yields a corresponding
integrality-gap bound), which is a significant improvement over the previous-best 
$O(\log n)$-approximation ratio obtained by~\cite{BockGKS11}. 
Although \eqref{minklp} has an exponential number of variables, one can obtain a
near-optimal solution $x^*$ by solving the dual LP (which has an exponential number of
constraints) to near-optimality, which can be achieved by using an approximation 
algorithm for orienteering to obtain an approximate separation oracle for the dual. We 
prove the following lemma in Section~\ref{lpsolve}.

\begin{lemma}\label{lem:lp} \label{klplem}
We can use a $\gamma_{\morient}$-approximation algorithm for orienteering to efficiently
compute a feasible solution $x^*$ to \eqref{minklp} of value at most
$\gamma_{\morient}\cdot\OPT$. 
\end{lemma}

Let $k^*=\sum_P x^*_P$ (so $k^*\leq\gamma_{\morient} \cdot \OPT$).  
Our goal is to round $x^*$ to a solution using at most $O(k^*)$ paths that have average 
regret $O(R)$. We can then apply Lemma~\ref{avg2max} to obtain $O(k^*)$ paths, each having
regret at most $R$, and thereby obtain an $O(1)$-approximate solution. 
We prove the following theorem.

\begin{theorem}\label{thm:schoolbus} \label{minkapx}
We can efficiently round $x^*$ to a solution using at most
$(8 + 4\sqrt 3)k^*+1$ rooted paths. 
This yields $(8+4\sqrt{3})\gm_{\morient}\OPT+1\leq 30.86\cdot\OPT$ rooted paths by taking 
$\gm_{\morient}=2+\e$~\cite{ChekuriKP12}, and shows that the integrality gap of
\eqref{minklp} is at most $9+4\sqrt{3}\leq 15.93$. 
\end{theorem}

We first give an overview on the rounding procedure 
that obtains a slightly worse approximation ratio. We show in 
Section~\ref{improve} how to refine this to obtain the guarantee stated above.
Let $\supp(x^*):=\{P: x^*_P > 0\}$ be the paths in the support of $x^*$.
To gain some intuition, suppose
first that it happens that when we direct every path $P\in\supp(x^*)$ away from $r$, we
obtain a directed graph $H$ that is acyclic. We can then set up a network-flow 
problem to find a minimum $c^{\reg}$-cost flow in $H$ of value at most $\ceil{k^*}$ such
that every node has at least one unit of flow entering it. Since $x^*$ can be viewed as a
path decomposition of a feasible
flow of $c^{\reg}$-cost at most $k^* R$, by the integrality property of flows, there is an
integral flow of $c^{\reg}$-cost at most $k^* R$. Since $H$ is acyclic, this flow may be
decomposed into at most $\ceil{k^*}$ paths that cover all the nodes, and the average
regret of this path collection is at most $R$, so we obtain the desired rounding.

Of course, in general $H$ will not be acyclic and rounding $x^*$ as above may yield an
integral flow that does not decompose into a collection of only paths. So we seek to identify a
subset $W\sse V$ of ``witness'' nodes and a collection of $O(k^*)$ fractional paths from
$\C_R$ covering $W$ such that: 
(a) directing each path in this collection away from $r$ yields a DAG; and
(b) given any collection of integral paths covering $W$, one can graft the nodes of $V\sm W$
into these paths (to obtain new paths covering $V$) incurring an additional $c^{\reg}$-cost of
$O(k^* R)$. 
Property (a) allows one to use the aforementioned network-flow argument to 
obtain $O(k^*)$ paths covering $W$ with total regret $O(k^*R)$, and property (b) enables
one to modify this to obtain $O(k^*)$ (integral) paths covering $V$ while keeping the
total regret to $O(k^*R)$ (so that one can then apply Lemma~\ref{avg2max}).  
 
To obtain $W$, we carefully construct a forest $F$ of cost $O(k^* R)$ (step A1 below) with
the property that for every component $Z$ of $F$, we can associate a single node $w\in Z$,
which we include in $W$, such that there is a 
total $x^*$-weight of at least $0.5$ in paths $P$ containing $w$ for which 
$\red(w,P)\sse Z$.  
Notably, we achieve this in a rather clean and simple way by defining a 
{\em downwards-monotone cut-requirement function} based on the fractional solution $x^*$ 
that encodes the above requirement, an idea that we believe has wider applicability, especially 
for network-design problems.%
\footnote{We point out that the LP-rounding algorithms of~\cite{GuptaRS07,GuptaK09}
for {\em stochastic Steiner tree} problems use the LP-solution to guide a primal-dual 
process for constructing a suitable forest, which
is in fact precisely the primal-dual process of~\cite{AgrawalKR95,GoemansW95} applied to a
suitable cut-requirement function. 
By making this function explicit, we obtain a more illuminating explanation for the
algorithm and a simpler, cleaner description and analysis.} 

Once we have such a forest, property (b) holds by construction since the total cost of
$F$ is $O(k^* R)$ (Lemma~\ref{lem:forest}). Moreover (step A2), if we shortcut each path
$P\in\supp(x^*)$ so that it only contains nodes $w\in W$ for which $\red(w,P)$ is
contained in some component of $F$, then the resulting paths cover each node in $W$ to an
extent of at least $0.5$ and satisfy the conditions of part (ii) of Lemma~\ref{dinc} 
(see Lemma~\ref{witpaths}). So by doubling the fractional values of the resulting paths,
we obtain a fractional-path collection satisfying property (a). Hence, we can obtain
$O(k^*)$ integral paths covering $W$ (step A3) and attach the nodes of $V\sm W$ to these
paths (step A4) while ensuring that the total regret remains $O(k^* R)$
(Lemma~\ref{lem:patching}), and then apply Lemma~\ref{avg2max}. 
We prove in Theorem~\ref{apx1} that the resulting solution uses at most $16k^*+1\leq
16\gm_{\morient}\cdot\OPT+1$ paths. 
In Section~\ref{improve}, we show how to obtain the improved guarantee stated in
Theorem~\ref{minkapx} by fine-tuning the threshold used to form the forest $F$.
We now describe the algorithm in detail and proceed to analyze it.

{\small \vspace{5pt}
\begin{algorithm} \label{mainalg}
\vspace{-5pt} \hrule \vspace{5pt}

\noindent
Input: A fractional solution $x^*$ to \eqref{minklp} obtained via Lemma~\ref{klplem};
$k^*=\sum_P x^*_P$.

\noindent 
Output: $O(k^*)$ paths, each having regret at most $R$, covering all the nodes.
\end{algorithm}
\begin{labellist}
\item {\bf Finding a low-cost forest \boldmath $F$.\ } For a subset 
$S\sse V\cup\{r\}$ and a node $v$, define $\tau(v,S):=\sum_{P:\red(v,P)\sse S}x^*_P$; 
define $f(S)=1$ if $\tau(v,S)<\frac{1}{2}$ for all $v\in S$, and 0 otherwise. Note that
$f$ is a {\em downwards-monotone} cut-requirement function: if $\es\neq A\sse B$ then 
$f(A)\geq f(B)$. We call a set $S$ with $f(S)=1$, an {\em active} set.
\begin{list}{A\arabic{enumi}.\arabic{enumii}}{\topsep=0.5ex \itemsep=0ex
    \usecounter{enumii}}
\item Use the 2-approximation algorithm for $\{0,1\}$ downwards-monotone functions
in~\cite{GoemansW94} to obtain a forest $F$ such that $|\dt(S)\cap F|\geq f(S)$ for every
set $S\sse V\cup\{r\}$. 

\item For every component $Z$ of $F$ with $r\notin Z$, choose a node $w\in Z$ such 
that $\tau(w,Z)\geq\frac{1}{2}$ (which exists since $f(Z)=0$). 
Call $w$ the {\em witness node} for $Z$, and denote $Z$ by $Z_w$. 
Obtain a tour $h(Z)$ traversing all nodes of $Z$ by doubling the edges of $Z$ and
shortcutting. Let $W\sse V$ be the set of all witness nodes. 
\end{list}

\item {\bf Obtaining a fractional acyclic flow covering \boldmath $W$.\ }
\begin{list}{A\arabic{enumi}.\arabic{enumii}}{\topsep=0.5ex \itemsep=0ex
    \usecounter{enumii}}
\item For every path $P\in\supp(x^*)$ we do the following. 
Let $P_W\sse P\cap W$ be the set of witness nodes $w\in P$ such that $\red(w,P)$ is contained
in $Z_w$. 
We shortcut $P$ past the nodes in $P\sm(P_W\cup\{r\})$ to obtain a
rooted path $\phi(P)$ spanning the nodes in $P_W$. Note that 
shortcutting does not increase the $c^{\reg}$-cost. 
Let $\C'\sse\C_R=\{\phi(P): P\in\supp(x^*), \phi(P)\neq\{r\}\}$ denote this new
collection of non-trivial paths. 

\item Let $H=(\{r\}\cup V,A_H)$ be the directed graph obtained by directing each path in
$\C'$ away from $r$. Let $z=(z_a)_{a\in A_H}$ be the flow that sends 
$\sum_{P:\phi(P)=P'}x^*_P$ flow along each path $P'\in\C'$. 
We prove in Lemma~\ref{witpaths} that $H$ is acyclic, and that
$z^{\into}(w):=\sum_{a\in\dt_H^{\into}(w)}z_a\geq\frac{1}{2}$ for every $w\in W$. 
\end{list}

\item Use the integrality property of flows to round $2z$ to an integer flow $\hz$ of no
greater $c^{\reg}$-cost and value $k\leq\ceil{2k^*}$ such that $\hz^{\into}(w)\geq 1$
for every $w\in W$. 
Since $H$ is acyclic, we may decompose $\hz$ into $k$ rooted paths
$\hP_1,\ldots,\hP_k$ so that (possibly after some shortcutting) every node of $W$ lies
on exactly one $\hP_i$ path.

\item {\bf Grafting in the nodes of \boldmath $V\sm W$.\ }
If there is a component $Z$ of $F$ containing $r$, we pick an arbitrary path, say $\hP_1$,
and modify $\hP_1$ by traversing $h(Z)$ first and then visiting the nodes of
$\hP_1\sm\{r\}$ (in the same order as $\hP_i$). Next, for every path
$\hP_i,\ i=1,\ldots,k$, we walk along $\hP_i$ and each time we visit a new node $w\in W$
on $\hP_i$ we traverse $h(Z_w)$ before moving on to the next node on $\hP_i$. Let $\tP_i$
denote the resulting new path. 

\item Apply Lemma~\ref{avg2max} to $\tP_1,\ldots,\tP_k$ to obtain the final set of paths
(having maximum regret $R$). 
\end{labellist}
\hrule}

\paragraph{Analysis.} Let $\comp(F)$ denote the set of components of $F$. Note that
$V\sse\bigcup_{Z\in\comp(F)} Z$.

\begin{lemma}\label{lem:forest}
The forest $F$ computed in step A1 has cost at most $6\cdot k^*\cdot R$. Thus,
$\sum_{Z\in\comp(F)}c(h(Z))\leq 12k^* R$.
\end{lemma}

\begin{proof}
Consider the following LP for covering the cuts $\dt(S)$ corresponding to active sets $S$.
\begin{equation}
\min \quad \sum_{e} c_ez_e \qquad 
\text{s.t.} \qquad z(\delta(S)) \geq 1 \quad \forall\es\subsetneq S \sse V\cup\{r\},
\qquad z \geq 0. \tag{C-P} \label{cutlp}
\end{equation}
Define $z$ by setting $z_e=\sum_{P: e\text{ is red on }P} 2\cdot x^*_P$ for all $e$.
This is a feasible solution to \eqref{cutlp} since for every active set $S$ and every
node $v\in S$, we have that 
\[\frac{1}{2}<1-\tau(v,S)\leq\sum_{P:\red(v,P)\not\sse S}x^*_P
\leq\sum_{e\in\dt(S)}\Bigl(\sum_{P:e\in\red(v,P)}x^*_P\Bigr)\leq z(\dt(S))/2.\] 
Also, $\sum_e c_ez_e=2\sum_Px^*_P\bigl(\sum_{\text{$e$ red on $P$}} c_e\bigr)\leq 
3\sum_P c^{\reg}(P)x^*_P\leq 3k^*R$. The penultimate inequality follows from 
Lemma~\ref{lem:redcost}, and the last inequality follows because 
$\supp(x^*)\sse\C_R$ and $\sum_P x^*_P=k^*$.
The 2-approximation algorithm of \cite{GoemansW94} then guarantees that 
$c(F)\leq 2\cdot\OPT_{\eqref{cutlp}}\leq 6k^*R$,
Since $c(h(Z))\leq 2c(Z)$ for component $Z$ of $F$, we have 
$\sum_{Z\in\comp(F)} c(h(Z))\leq 12k^*R$. 
\end{proof}

\begin{lemma} \label{lem:summary} \label{witpaths}
(i) For every path $P\in\C_R$, every red interval of $P$ contains at most one node of
$P_W$. Therefore, $\phi(P)$ visits nodes $v$ in strictly increasing order of $\dist_v$;
(ii) $\sum_{P:w\in\phi(P)}x^*_P\geq\frac{1}{2}$ for every $w\in W$;
(iii) Hence, the digraph $H$ constructed in step A2 is acyclic, and
$z^{\into}(w)\geq\frac{1}{2}$ for every $w\in W$. 
\end{lemma}

\begin{proof}
Part (iii) follows immediately from parts (i) and (ii).
For part (i), recall that $P_W=\{w\in P\cap W: \red(w,P)\sse Z_w\}$. 
If there are two nodes $u$, $w$ of $P_W$ contained in some red interval of $P$ then
$Z_u\cap Z_w\neq\es$, but this contradicts the fact that we add at most one node to $W$ 
from each component of $F$. 
It follows that $\phi(P)$ contains at most one node from each red interval of $P$, and
by Lemma~\ref{lem:cross}, we have that $\phi(P)$ visits nodes $v$ in strictly increasing
order of distance $\dist_v$. 
For part (ii), we note that for a node $w\in W$, by definition, we have that $w\in\phi(P)$ 
iff $\red(w,P)\sse Z_w$. So 
$\sum_{P:w\in\phi(P)}x^*_P=\sum_{P:\red(w,P)\sse Z_w}x^*_P=\tau(w,Z_w)\geq\frac{1}{2}$,
where the last inequality follows from the definition of $Z_w$.
\end{proof}

\begin{lemma}\label{lem:patching}
The total regret of the paths $\tP_1,\ldots,\tP_k$ obtained in step A4 is at most 
$14\cdot k^* \cdot R$. 
\end{lemma}

\begin{proof} 
Let $Z_r$ denote the component of $Z$ containing $r$; let $Z_r=\es$ and $h(Z_r)=0$ if
there is no such component. 
The regret of path $\tP_1$ is 
$c^{\reg}(\hP_1)+\sum_{w\in\hP_1} c^{\reg}\bigl(h(Z_w)\bigr)$ and the regret of $\tP_i$
for $i\neq 1$ is $c^{\reg}(\hP_i)+\sum_{w\in\hP_i:w\neq r} c^{\reg}\bigl(h(Z_w)\bigr)$.  

The paths $\hP_1,\ldots,\hP_k$ 
obtained after step A3 have $c^{\reg}$-cost at most the $c^{\reg}$-cost of $2z$, which is
\linebreak 
$2\sum_{P'\in\C'}c^\reg(P')\sum_{P:\phi(P)=P'}x^*_P\leq\sum_Pc^\reg(P)x^*_P\leq 2k^*R$.
Since the $c^\reg$- and $c$-costs of a cycle are identical (by Fact~\ref{obs:regret}), the
total regret of $\tP_1,\ldots,\tP_k$ is at most 
$\sum_{i=1}^k c^\reg(\hP_i)+\sum_{Z\in\comp(F)} c\bigl(h(Z)\bigr)\leq
2k^*R+12k^*R=14k^*R$, where we use Lemma~\ref{lem:forest} for the last
inequality. 
\end{proof}

\begin{theorem} \label{apx1}
Algorithm~\ref{mainalg} returns a feasible solution with at most 
$16k^*+1\leq 16\gm_{\morient}\cdot\OPT+1$ paths.
\end{theorem}

\begin{proof}
Applying Lemma~\ref{avg2max} to the paths $\tP_1,\ldots,\tP_k$, which have total regret at
most $14k^* R$ (by Lemma~\ref{lem:patching}), we obtain a collection of $k'$ rooted paths of
maximum regret $R$ whose union covers all nodes, where 
$k'\leq\bigl(\frac{14k^*}{k}+1\bigr)k\leq 14k^*+k\leq 14k^*+\ceil{2k^*}\leq 16k^*+1$.
\end{proof}

\subsection{Improvement to the guarantee stated in Theorem~\ref{minkapx}} 
\label{sec:optimize} \label{improve}
We now describe the improvement that yields Theorem~\ref{minkapx}.
Let $\dt\in(0,1)$ be a parameter that we will fix later. 
The only change is that we now define the cut-requirement function in step A1 as $f(S)=1$
if $\tau(v,S)<\delta$ for all $v\in S$, and 0 otherwise.
This results in a corresponding change to the integer flow $\hz$ obtained in step A2.  

Mimicking the proof of Lemma~\ref{lem:forest}, we see that setting 
$z_e=\sum_{P:\text{$e$ is red on $P$}}x^*_P/(1-\dt)$ yields a feasible solution to
\eqref{cutlp}, and therefore we have, $c(F)\leq\frac{3}{1-\delta}\cdot k^*R$, and
$\sum_{Z\in\comp(F)}c(h(Z))\leq\frac{6}{1-\delta} \cdot k^*R$.
Step A2 is unchanged, but parts (ii) and (iii) of Lemma~\ref{witpaths} need to be suitably
modified: we now have that the flow $z$ satisfies $z^{\into}(w)\geq\dt$ for
every $w\in W$. Correspondingly, we round $z/\dt$ to an integer flow in step A3, and
obtain at most $\ceil{k^*/\dt}$ paths. 
Proceeding as in the proof of Lemma~\ref{lem:patching}, we infer that the total regret of
the paths obtained after grafting in the nodes of $V\sm W$ is at most
$\bigl(\frac{1}{\dt}+\frac{6}{1-\dt}\bigr)k^*R$. 

Applying Lemma~\ref{avg2max}, this yields at most 
$\bigl(\frac{1}{\dt}+\frac{6}{1-\dt}\bigr)k^*+\left\lceil\frac{k^*}{\dt}\right\rceil
\leq\bigl(\frac{2}{\dt}+\frac{6}{1-\dt}\bigr)k^*+1$ paths having regret at most $R$. 
Taking $\dt=\frac{\sqrt{3}-1}{2}$ to minimize the coefficient of $k^*$, we obtain the
guarantee stated in Theorem~\ref{minkapx}.

\section{Multiplicative-{\large \nrvrp}} \label{multreg}
Recall that in multiplicative-\rvrp, we are given a regret-bound $R$, and we want to find
the minimum number of paths covering all nodes so that each node $v$ is visited by time
$R\cdot\dist_v$. When $R=1$, the problem can be solved in polytime (as this is simply
additive-\rvrp with regret-bound 0), so we assume that $R>1$.
We show that multiplicative-\rvrp reduces to \rvrp incurring an
$O\bigl(\log(\frac{R}{R-1}\bigr)$-factor loss. 
The following observation, which falls out of Lemma~\ref{avg2max} will be quite useful. 

\begin{lemma} \label{covering}
Let $\gm_{\mrvrp}$ be the approximation ratio of our \rvrp-algorithm.
Suppose there are $k$ paths covering a given set $S$ of nodes ensuring that every node in
$S$ has additive regret at most $\rho$. For any $\e>0$, one can efficiently
obtain at most $\floor{\gm_{\mrvrp}k\ceil{\frac{1}{\e}}}$ 
\nolinebreak
\mbox{paths covering $S$ such that each node in $S$ has regret at most $\e\rho$.}
\end{lemma}

\begin{proof} 
We shortcut the $k$ paths so that they only contain the nodes in $S$. The regret of each
of these $k$ paths is at most $\rho$, so as in Lemma~\ref{avg2max}, we may
break up each path into at most $\ceil{\frac{1}{\e}}$ paths of regret at most
$\e\rho$.  This creates at most $k\ceil{\frac{1}{\e}}$ paths of regret at most $\e\rho$
that cover $S$. So by using our algorithm for \rvrp with the node-set $r\cup S$ and
regret-bound $\e\rho$, we obtain $\gm_{\mrvrp}k\ceil{\frac{1}{\e}}$ paths of regret at
most $\e\rho$ covering $S$. Since the number of paths is an integer, we
actually have $\floor{\gm_{\mrvrp}k\ceil{\frac{1}{\e}}}$ paths.
\end{proof}

\begin{theorem} \label{multrvrp}
Multiplicative-\rvrp can be reduced to additive-\rvrp incurring an  
$O\bigl(\log(\frac{R}{R-1})\bigr)$-factor loss.
This yields an $O\bigl(\log(\frac{R}{R-1})\bigr)$-approximation for multiplicative-\rvrp.  
\end{theorem} 

\begin{proof}
Let $R=1+\dt$.
For $i\geq 1$, define $V_i=\{v: 2^{i-1}\leq\dist_v<2^{i}\}$. Note that the $V_i$s partition
the non-root nodes. 
Let $\iopt$ denote the optimal value of the multiplicative-\rvrp instance. 
We apply Lemma~\ref{covering} with $\e=\frac{1}{4}$ to the $V_i$s:
for each $V_i$, there are $\iopt$ paths covering $V_i$ such that each node in $V_i$ has
regret at most $\dt\cdot 2^{i}$, 
so we obtain at most $N=\floor{4\gm_{\mrvrp}\iopt}=O(\iopt)$ paths covering $V_i$ such
that each node in $V_i$ has regret at most $\dt\cdot 2^{i-2}$. Pad these with the trivial
path $\{r\}$ if needed, to obtain exactly $N$ paths $P^i_1,\ldots,P^i_N$. 

Let $M=\ceil{\log_2\bigr(3+\frac{8}{\dt}\bigr)}=O\bigl(\log_2(\frac{R}{R-1})\bigr)$.
Now for every index $i=1,2,\ldots,M$ and every $j=1,\ldots,N$, we concatenate the
paths $P^i_j,P^{i+M}_j,P^{i+2M}_j, \ldots$ by moving from the end-node of $P^{i+aM}_j$ to
$r$ before following $P^{i+(a+1)M}_j$ for each $a\geq 0$. This yields
$MN$ paths that together cover all nodes. 

To finish the proof, we show that every node $v$ is visited by time $R\cdot\dist_v$.  
Suppose $v\in V_{i+aM}$ and is covered by path $P^{i+aM}_j$. It's visiting
time is then at most 
$c_{P^{i+aM}_j}(v)+2\sum_{a'<a} c(P^{i+a'M}_j)
\leq\dist_v+\dt\cdot 2^{i+aM-2}+2\bigl(1+\frac{\dt}{4}\bigr)\sum_{a'<a}2^{i+a'M}
\leq \dist_v\bigl(1+\frac{\dt}{2}+4\cdot\frac{1+\dt/4}{2^M-1}\bigr)
\leq R\cdot\dist_v$.
\end{proof}

\section{Applications to {\large \ndvrp}} \label{dvrp}
Recall that the goal in \dvrp is to find the fewest number of rooted paths of {\em length}
at most $D$ that cover all the nodes. We say that a rooted path $P$ is feasible if
$c(P)\leq D$. The length-$D$ prefix of a rooted path $P$, denoted by
$P(D)$, is the portion of $P$ starting from $r$ and up to (and including) the last node
$v\in P$ such that $c_P(v)\leq D$.
Let $\iopt$ be the optimal value, and $R_{\max}\leq D$ be the maximum regret of a path in
an optimal solution, which we can estimate within a factor of 2.

We apply our algorithm for \rvrp to prove two approximation results for \dvrp. We obtain
an $O\bigl(\frac{\log R_{\max}}{\log\log R_{\max}}\bigr)$-approximation algorithm
(Theorem~\ref{dvrpthm}), which also yields the same upper bound on the integrality gap of
a natural configuration LP \eqref{dvrplp} for \dvrp (Corollary~\ref{dvrpcor}).
This improves upon the $O(\log D)$ approximation factor and integrality-gap bound for
\eqref{dvrplp} proved in~\cite{NagarajanR08}.
Note that for {\em graphical metrics}, that is, when the underlying metric is the
shortest-path metric of an {\em unweighted} graph, we may assume that $D\leq n^2$, and so
this yields an $O\bigl(\frac{\log n}{\log\log n}\bigr)$-approximation.
Next, we show that the integrality gap of \eqref{dvrplp} is also at most $O(\optdvlp)$
(Theorem~\ref{optintgap}).  
This presents an interesting contrast with set cover for which there are 
$O(\log n)$-integrality-gap examples for the standard set-cover LP even when the optimal
LP-value is a constant. The configuration LP \eqref{dvrplp} is also a set-cover LP, but
our result precludes such an integrality-gap construction for this LP. 
Our integrality-gap bounds suggest that the additional structure in \dvrp can be further
exploited, perhaps by refining our methods, to derive improved guarantees.

\subsection{An \boldmath $O\bigl(\frac{\log R_{\max}}{\log\log R_{\max}}\bigr)$-approximation}
\label{dvrpapprox}
As a warm-up, note that a simple
$O(N)$-approximation, where $N=\ceil{\log_2(\frac{R_{\max}}{D-\max_{v:\dist_v<D}\dist_v})}$,
follows by applying Lemma~\ref{covering} with $\e=\frac{1}{2}$ to the node-sets  
$V_0=\bigl\{v: D-\dist_v\geq\frac{R_{\max}}{2}\bigr\}$, 
$V_i=\bigl\{v: D-\dist_v\in[\frac{R_{\max}}{2^{i+1}},\frac{R_{\max}}{2^i})\bigr\}$ for
$i=1,\ldots,N-1$, and $V_N=\{v: \dist_v=D\}$, which partition $V$.  
For $V_N$, we can obtain at most $\iopt$ regret-0 paths covering it.
Each $V_i,\ 0\leq i<N$, the optimal solution uses $\iopt$ to cover $V_i$ causing regret at
most $\frac{R_{\max}}{2^i}$ to these nodes. So, we obtain $O(\iopt)$ paths covering $V_i$ 
causing regret at most $\frac{R_{\max}}{2^{i+1}}$ to the $V_i$ nodes; 
hence, the length-$D$ prefixes of these paths cover $V_i$.

We now describe a more-refined reduction yielding  
an improved $O\bigl(\frac{\log R_{\max}}{\log\log R_{\max}}\bigr)$-approximation.  
The algorithm is again based on choosing suitable pairs of regret bounds and node-sets,
and covering each node-set using paths of the corresponding regret bound. However, instead
of {\em fixing} the regret bounds to be $\frac{R_{\max}}{2^i}$, we now obtain
them by solving a dynamic program (DP). 

Let $S_i=\{v: D-\dist_v<2^i\}$ for $i=0,\ldots,M=\ceil{\log_2 D}$. We use DP to obtain a
set of feasible paths $\Pc(i)$ covering $S_i$ for all $i$.
We use $F(i)$ to denote $|\Pc(i)|$. 
For all $0\leq k<i$, we use our algorithm for \rvrp to find a collection $\Qc(i,k)$ of
paths of regret at most $2^k$ that cover $S_i$, 
Let $\Pc(0)$ be the fewest number of paths of regret 0 (and hence are feasible) that cover
the nodes with $\dist_v=D$, which we can efficiently compute. 
For $i>0$, we set $F(i)=\min_{0\leq k<i}\bigl(|\Qc(i,k)|+F(k)\bigr)$; if $k'$ is the index
that attains the minimum, then we set 
$\Pc(i)=\Pc(k')\cup(\text{length-$D$ prefixes of the paths in $\Qc(i,k')$})$. 
We return the solution $\Pc(M)$, which we show is feasible.

\paragraph{Analysis.}
For ease of comprehension, we prove the approximation guarantee with respect to
an integer optimal solution here, and observe in Section~\ref{intgapbnds} that the
analysis easily extends to yield the same guarantee with respect to the optimum value of
the configuration LP. 

Let $\gm=\gm_{\mrvrp}<31$ be the approximation ratio of our \rvrp-algorithm.
We define a suitable set of indices, that is, regret bounds,
such that using these indices in the DP yields the desired bound on the number of paths. 
In order to establish a bound on $F(i)$ by plugging in a suitable index $k<i$ we need two
things. 
First, we need to bound $|\Qc(k,i)|$. 
This requires a more sophisticated analysis than the one suggested by
Lemma~\ref{covering}. Instead of directly using all the paths from an optimal solution  
to bound the number of paths of certain regret required to cover a given set of nodes, we
proceed as follows. We argue that by suitably preprocessing the paths in an optimal
solution (see Claim~\ref{pmodify}), we can obtain a near-optimal solution $\Oc$ such that
$S_i$ is covered by {\em paths of $\Oc$ of regret at most $2^i$}. We modify these
$\Oc$-paths by breaking them up (as in Lemma~\ref{avg2max}) to obtain paths of regret at
most $2^k$ that cover $S_i$, which yields a bound on $|\Qc(i,k)|$. 
Second, we need to argue that we make suitable progress when moving from index $i$ to index
$k$. In a crucial departure from the previous analysis, we make progress by either
suitably decreasing the {\em number}, or the maximum regret, of the paths, needed
from $\Oc$ to cover the remaining set of nodes. These ingredients yield the following
theorem. 

\begin{theorem} \label{dvrpthm}
$F(M)\leq O\bigl(\frac{\log R_{\max}}{\log\log R_{\max}}\bigr)\cdot\gm_{\mrvrp}N$. So the
above algorithm is an $O\bigl(\frac{\log R_{\max}}{\log\log R_{\max}}\bigr)$-approximation
algorithm for \dvrp (where $R_{\max}\leq D$ is the maximum regret of a path in an optimal
solution). 
\end{theorem}

We start with the following simple, but useful claim.

\begin{claim} \label{pmodify}
Let $P$ be a rooted path. We can obtain at most two paths $P_1$ and $P_2$, both ending at
$v=\arg\max_{v\in P}\dist_v$ and together covering all the nodes on $P$ such that
$c^{\reg}(P_1), c^{\reg}(P_2)\leq c^{\reg}(P)$ and $c(P_1), c(P_2)\leq c(P)$.
\end{claim}

\begin{proof} 
Let $t$ be the end-node of $P$. 
Let $P_1$ be the $r\leadsto v$ portion of $P$, and $P_2$ be the path where we move from
$r$ to $t$, and then traverse the $t\leadsto v$ portion of $P$.
Clearly, $c^{\reg}(P_1)\leq c^{\reg}(P)$ and $c(P_1)\leq c(P)$. 
Also, $c(P_2)=\dist_{t}+c(P)-c(P^1)\leq\dist_{v}+c(P)-c(P_1)\leq c(P)$ and
$c^{\reg}(P_2)=c(P_2)-\dist_v\leq c(P)-\dist_v\leq c(P)-\dist_t=c^{\reg}(P_2)$. 
\end{proof}

We preprocess the paths in an optimal solution using Claim~\ref{pmodify} losing a factor
of 2. (Note that this is solely for the purposes of analysis.)
Let $\Oc=\{P^*_1,\ldots,P^*_N\}$ denote the resulting collection of paths, where 
$c^{\reg}(P^*_1)\leq\ldots\leq c^{\reg}(P^*_N)=R_{\max}$, and $N\leq 2\iopt$.
For $i\geq 0$, define $n(i)=|\{P\in\Oc: c^{\reg}(P)<2^i\}|$.
The preprocessing ensures that if a node $v\in S_i$ lies on a path $P\in\Oc$, then 
$c^{\reg}(P)\leq D-\dist_v<2^i$. 
So $S_i$ is covered by the $n(i)$ paths of $\Oc$ of regret less than $2^i$.

\begin{lemma} \label{feas}
For all $i$, $\Pc(i)$ consists of feasible paths that cover $S_i$.
\end{lemma}

\begin{proof} 
The proof is by induction on $i$. The base case when $i=0$ clearly holds.
For $i>0$, $\Pc(i)=\Pc(k)\cup(\text{length-$D$ prefixes of the paths in $\Qc(i,k)$})$ for
some $k<i$. By the induction hypothesis, the paths in $\Pc(k)$ have length at most $D$ and
cover $S_k$. So all paths in $\Pc(i)$ have length at most $D$. Consider a node 
$v\in S_i\sm S_k$, and suppose that $v$ lies on path $P\in\Qc(i,k)$. Then,
$c_P(v)=c^{\reg}(P)+\dist_v\leq 2^k+\dist_v\leq D$, where the last inequality follows
since $v\notin S_k$ implies that $D-\dist_v\geq 2^k$. Thus, $v$ is covered by $P(D)$.
\end{proof}

\begin{claim} \label{helper}
$F(z)\leq 2^z\gm n(z)$ for all $z$.
\end{claim}

\begin{proof}
This is clearly true for $z=0$. For $z\geq 1$, $F(z)\leq|\Qc(z,0)|+F(0)$.
We have $|\Qc(z,0)|\leq\gm\bigl(n(1)+2^z(n(z)-n(1))\bigr)$ since the $n(z)-n(1)$ paths of
$\Oc$ of regret more than 1 can be broken up to yield at most $2^z$ (in fact $2^{z-1}$)
paths of regret at most 1. Thus, $F(z)\leq 2^z\gm n(z)$ (note that $\gm\geq 1$). 
\end{proof}

We will need the following technical lemma, whose proof we defer to the end of this
subsection. 

\begin{lemma} \label{sequence}
Let $\{a_i\}_{i=0}^k$ be a sequence of integers such that $a_0>a_1>\ldots>a_k>0$, and 
$2^{a_{i+1}-a_i}<\frac{\log_2 a_{i+1}}{a_{i+1}}$ for all $i=0,\ldots,k-1$. Then,
$k=O\bigl(\frac{a_0}{\log a_0}\bigr)$.
\end{lemma}

\begin{proofof}{Theorem~\ref{dvrpthm}}
Let $k(0)=0$.
For $i>0$, set $k(i)=0$ if $\bigl(n(i)-n(0)\bigr)\cdot 2^i<n(i)$; otherwise, 
let $k(i)$ be the unique value of $k\in\{0,\ldots,i-1\}$ such that 
$\bigl(n(i)-n(k)\bigr)\cdot 2^{i-k}\geq n(i)>\bigl(n(i)-n(k+1)\bigr)\cdot 2^{i-k-1}$
(which must exist). 
This choice of index in the expression for $F(i)$ (roughly speaking) is tailored to ensure 
that $|\Qc(i,k(i))|\leq 3\gm n(i)$ {\em and} 
$n\bigl(k(i)\bigr)\leq n(i)\bigl(1-2^{k(i)-i}\bigr)$ when $k(i)>0$.
To see the bound on $|\Qc(i,k(i))|$, consider an arbitrary index $0\leq k<i$.
There are $n(i)-n(k+1)$ paths in $\Oc$ with regret in the range $[2^{k+1},2^i)$, and
$n(k+1)-n(k)$ paths in $\Oc$ with regret in the range $[2^k,2^{k+1})$. 
Breaking up these paths into paths of regret $2^k$ as in Lemma~\ref{avg2max}, and
combining with the $n(k)$ paths of $\Oc$ of regret less than $2^k$  
yields at most $n(k)+2\bigl(n(k+1)-n(k)\bigr)+\bigl(n(i)-n(k+1)\bigr)(1+2^{i-k-1})
\leq 2n(i)+\bigl(n(i)-n(k+1)\bigr)2^{i-k-1}$
paths of regret at most $2^k$ covering $S_i$. 
So by Lemma~\ref{covering}, we have 
$$
|\Qc(i,k(i))|\leq\gm\bigl[2n(i)+\bigl(n(i)-n(k(i)+1)\bigr)2^{i-k(i)-1}\bigr]
< 3\gm n(i).
$$

We now define a function $G:\{0,\ldots,M\}\mapsto\Z_+$ based on the recurrence for
$F$ by plugging in index $k=k(i)$ in the definition of $F(i)$.
More precisely, set $G(i)=F(i)$ for $i\leq 3$.
Set $G(i)=|\Qc(i,k(i))|+G\bigl(k(i)\bigr)$ for $i>3$. It is easy to see by induction that
$F(i)\leq G(i)$ for all $i$. 

Let $M'=1+\floor{\log_2(R_{\max})}$. So $2^{M'}>R_{\max}$, and $n(M')=N$.
If $M'<M$, then $F(M)\leq|\Qc(M,M')|+F(M')$, and $|\Qc(M,M')|\leq\gm\iopt$ since there
$\iopt$ paths of regret at most $2^{M'}$ that cover $V$.
So we have $F(M)\leq\gm N+F(M')\leq\gm N+G(M')$. To finish the proof, we prove
that $G(M')$ satisfies the bound in the theorem statement.

Let $k^{(0)}(i)=i$ and $k^{(j+1)}(i)=k\bigl(k^{(j)}(i)\bigr)$. 
Define 
$B(i)=\bigl\{j>0: k^{(j)}(i)>0,\ 2^{k^{(j)}(i)-k^{(j-1)}(i)}<\frac{\log_2 k^{(j)}(i)}{k^{(j)}(i)}\bigr\}$ 
and $b(i)=|B(i)|$.
Recall that $n\bigl(k(i)\bigr)\leq n(i)\bigl(1-2^{k(i)-i}\bigr)$ when $k(i)>0$.
Thus, when $k(i)>0$, if $b(i)=b\bigl(k(i)\bigr)$ then we decrease the number of paths of
$\Oc$ required to cover $S_{k(i)}$ by an appropriate factor, and otherwise, we decrease the
maximum regret of these paths considerably. 
Thus, $b(i)$ is a measure of the number of times we make progress starting from $i$ by
decreasing the maximum regret significantly. 
Let $g(x)=2$ for $x<4$ and $\frac{x}{\log_2 x}$ otherwise. Note that $g$ is an
increasing function.

We prove by induction on $i$ that $G(i)\leq 3\gm n(i)\cdot\bigl(g(i)+b(i)\bigr)+8\gm n(3)$. 
The base case is $i=0$, which holds since $G(0)=F(0)\leq n(0)$, since we can compute the
minimum number of regret-0 paths covering $S_0$.
Consider $i>0$. We have $G(i)\leq 3\gm n(i)+G\bigl(k(i)\bigr)$. If $k(i)\leq 3$, 
then $G(i)\leq 3\gm n(i)+F\bigl(k(i)\bigr)\leq 
3\gm n(i)+2^{k(i)}\gm n\bigl(k(i)\bigr)\leq 3\gm n(i)+8\gm n(3)$, where the second
inequality follows from Claim~\ref{helper}. 
So suppose $k(i)\geq 4$, and so $g\bigl(k(i)\bigr)=\frac{k(i)}{\log_2 k(i)}$.
Then
\begin{equation*}
\begin{split}
G(i) & \leq 3\gm n(i)+G\bigl(k(i)\bigr)
\leq 3\gm\biggl[n(i)
+n\bigl(k(i)\bigr)\Bigl(\tfrac{k(i)}{\log_2 k(i)}+b\bigl(k(i)\bigr)\Bigr)\biggr]+8\gm n(3) \\
& \leq 3\gm n(i)
\biggl[1+\bigl(1-2^{k(i)-i}\bigr)\cdot\tfrac{k(i)}{\log_2 k(i)}+b\bigl(k(i)\bigr)\biggr]
+8\gm n(3).
\end{split}
\end{equation*}
If $2^{k(i)-i}\geq\frac{\log_2 k(i)}{k(i)}$, then 
the above term is at most $3\gm n(i)\bigl[\frac{k(i)}{\log_2 k(i)}+b(i)\bigr]+8\gm n(3)$. 
Otherwise, $b(i)=b\bigl(k(i)\bigr)+1$, and the above term is bounded by
$3\gm n(i)\bigl[\frac{k(i)}{\log_2 k(i)}+b(i)\bigr]+8\gm n(3)$. 
Since $\frac{k(i)}{\log_2 k(i)}=g\bigl(k(i)\bigr)\leq g(i)$, we have 
$G(i)\leq 3\gm n(i)\cdot\bigl(g(i)+b(i)\bigr)+8\gm n(3)$, which completes the
induction step.
So by induction, $G(M')\leq 3\gm N\cdot\bigl(g(M')+b(M')\bigr)+8\gm N$.

For any $i>0$, consider the sequence $i,\{k^j(i)\}_{j\in B(i)}$. This sequence has
$b(i)+1$ terms and satisfies the conditions of Lemma~\ref{sequence}, so by
that lemma, $b(i)=O\bigl(\frac{i}{\log i}\bigr)$. 
Thus, 
$G(M')\leq 3\gm N\cdot O\bigl(\frac{M'}{\log M'}\bigr)
=\gm N\cdot O\bigl(\frac{\log R_{\max}}{\log\log R_{\max}}\bigr)$, which completes the
proof. 
\end{proofof}

\begin{proofof}{Lemma~\ref{sequence}}
We may assume that $a_k\geq 16$, since otherwise, we may truncate the sequence at $a_j$
where $a_{j+1}<16$; we have $k\leq j+a_{j+1}=O(j)$, so we can proceed to bound $j$.
Since $\log_2 x/x\leq 1/\sqrt{x}$ for $x\geq 16$, we have
$2^{a_{i+1}-a_i}<\frac{1}{\sqrt{a_{i+1}}}$ for all $i=0,\ldots,k-1$.

Let $i_0=0$. 
If $i_{j-1}\leq k$, let $i_j\leq k$ be
the smallest index such that $a_{i_j}<a_{i_{j-1}}/2$ if this exists; otherwise, let
$i_j=k+1$. 
Let $i_0,i_{j_1},\ldots,i_{j_\ell}\leq k<i_{j_{\ell+1}}$ be the indices generated this way. 
Note that $\ell\leq\log_2(a_0/a_k)\leq\log_2 a_0$.

Consider a subsequence $a_{i_j},\ldots,a_{i_{j+1}-1}$. We have
$2^{a_{i+1}-a_i}<\frac{1}{\sqrt{a_{i+1}}}\leq\sqrt{\frac{2}{a_{i_j}}}$ for all
$i=i_j,\ldots,i_{j+1}-2$. It follows that 
$2^{-a_{i_j}/2}\leq 2^{a_{(i_{j+1}-1)}-a_{i_j}}<\bigl(\frac{2}{a_{i_j}}\bigr)^{(i_{j+1}-i_j-1)/2}$
and hence, $i_{j+1}-i_j<1+f(a_{i_j})$, where $f(x)=\frac{x}{\log_2(x/2)}$. Note that $f$
is an increasing function. 

Adding the above inequality for all the $\ell+1$ subsequences
$\{a_i\}_{i=i_j}^{i_{j+1}-1}$ where $j=0,\ldots,\ell$, we obtain that 
$k+1\leq\ell+1+\sum_{j=0}^{\ell} f(a_{i_j})$, so 
$k\leq\ell+\sum_{j=0}^\ell f(a_0/2^j)$. Note that $f(x)$ is a concave function for
$x\in[16,\infty)$ and that $16\leq a_{i_\ell}<a_0/2^\ell$.
(We have $f'(x)=\frac{1}{\log_2 x-1}-\frac{1}{\ln 2(\log_2 x-1)^2}$ and
$f''(x)=-\frac{1}{x\ln 2(\log_2 x-1)^2}+\frac{2}{x\ln^2 2(\log_2 x-1)^3}
=\frac{1}{x\ln 2(\log_2 x-1)^2}\bigl(\frac{2}{\ln 2(\log_2 x-1)}-1\bigr)$.)
So 
$$
\sum_{j=0}^\ell f(a_0/2^j)
\leq (\ell+1)f\Bigl(\tfrac{\sum_{j=0}^{\ell}a_0/2^j}{\ell+1}\Bigr)
\leq (\ell+1)f\Bigl(\tfrac{2a_0}{\ell+1}\Bigr)
\leq \frac{2a_0}{\log_2\bigl(\frac{a_0}{\ell+1}\bigr)}
\leq \frac{8a_0}{\log_2 a_0}.
$$
The last inequality follows since
$\log_2\bigl(\frac{a_0}{\ell+1}\bigr)
\geq\log_2\bigl(\frac{a_0}{\log_2 a_0+1}\bigr)\geq\frac{\log_2 a_0}{4}$ since 
$a_0\geq 16$. 
Hence, $k\leq\log_2 a_0+\frac{8a_0}{\log_2 a_0}\leq\frac{9a_0}{\log_2 a_0}$ provided
$a_k\geq 16$. In general, we have
$k\leq 16+\frac{9a_0}{\log_2 a_0}=O\bigl(\frac{a_0}{\log a_0}\bigr)$.
\end{proofof}

\subsection{LP-based approximation guarantees for {\normalsize \ndvrp}} \label{intgapbnds} 
Consider the following configuration LP for \dvrp, which is along the same lines as 
\eqref{minklp}. Let $\Pc_D$ denote the collection of rooted paths of length at most $D$,
which is indexed below by $P$.
\begin{equation}
\min \quad \sum_{P} x_P \qquad 
\text{s.t.} \qquad \sum_{P: v \in P} x_P \geq 1 \quad \frall v\in V, \qquad 
x_P \geq 0 \quad \frall P. \tag{DV-P} \label{dvrplp}
\end{equation}

The preprocessing described in Claim~\ref{pmodify} 
can also be applied to a fractional solution to \eqref{dvrplp} losing a factor of
2. The only change is that when we create two paths $P_1, P_2$ from a path $P$ with
positive weight $x$, we increase the weights of $P_1$ and $P_2$ by $x$ and set the
(new) weight of $P$ to 0.
We break ties while preprocessing using an arbitrary, but fixed ordering over nodes; that
is, if $u, v$ are such that $\dist_u=\dist_v$ and $u$ comes before $v$ in the ordering,
\nolinebreak
\mbox{then we ensure that no fractional path ending at $v$ contains $u$.}

Thus, notice that the DP-based algorithm described and analyzed in
Section~\ref{dvrpapprox} 
also proves an integrality gap of $O\bigl(\frac{\log R_{\max}}{\log\log R_{\max}}\bigr)$,
where $R_{\max}\leq D$ is now the maximum regret of a path in the support of an optimal
LP-solution. The only change to the {\em analysis} is that the path-collection $\Oc$ is
(of course) replaced with the preprocessed LP-optimal solution, and that $n(i)$ is now
defined to be the total LP-weight of the fractional paths having regret less than $2^i$.
Nodes in $S_i$ are now covered to an extent of at least 1 by (preprocessed)
fractional paths having regret less than $2^i$. Also, since our algorithm for \rvrp has an
LP-relative guarantee, to bound $|\Qc(i,k)|$ it suffices to exhibit a collection of
fractional paths of regret at most $2^k$ that cover $S_i$, and this is done {\em exactly}
as before.  

\begin{corollary} \label{dvrpcor}
The integrality gap of \eqref{dvrplp} is 
$O\bigl(\frac{\log R_{\max}}{\log\log R_{\max}}\bigr)$,  
where $R_{\max}\leq D$ is the maximum regret of a path in the support of an optimal
solution to \eqref{dvrplp}.
\end{corollary}

Let $\optdvlp$ denote the optimal value of \eqref{dvrplp}. 
We now present an LP-rounding algorithm showing an integrality gap of $O(\optdvlp)$.
As with \eqref{minklp}, the separation problem for the dual of \eqref{dvrplp} is an
orienteering problem, so one can obtain a $\gm_{\morient}$-optimal solution to
\eqref{dvrplp} given a $\gm_{\morient}$-approximation algorithm for orienteering. 
Let $x^*$ denote the fractional solution obtained after preprocessing the
$\gm_{\morient}$-optimal solution to \eqref{dvrplp} as described above. Let 
$k^*=\sum_P x^*_P\leq 2\gm_{\morient}\optdvlp$. Let $\Pc'=\{P: x^*_P>0\}$.

\begin{theorem} \label{optintgap}
We can efficiently round $x^*$ to a feasible \dvrp-solution that uses $O({k^*}^2)$ paths.
\end{theorem}

\begin{proof}
For each $v \in V$ let $\mathcal P_v = \{P \in \mathcal P': \text{$P$ ends at $v$}\}$
and $B(v) = \{u \in V : \sum_{P \in \mathcal P_v: u \in P} x^*_P\geq\frac{1}{3k^*}\}$. 
We partition $V$ into $V_1,V_2,\ldots$ 
as follows. 
Take $v_1$ to be the node furthest from the root $r$, where we break ties using the same
ordering that was used in the preprocessing step, and set $V_1=B(v_1)$. 
In general, suppose we have formed $V_1,\ldots,V_{i-1}$, 
We pick $v_{i+1}$ to be node in $V \setminus \bigcup_{j=1}^{i-1} V_j$ that is furthest
from $r$ (breaking ties as before), and set $V_{i} = B(v_{i}) \setminus \cup_{j=1}^{i-1}
V_j$. It is clear that the $V_i$s are disjoint; we prove that $v_i \in V_i$ for every $i$,
and hence the $V_i$s form a partition of $V$.

For each $i$ we have $\sum_{P: v_i \in P} x^*_P\geq 1$. 
Since we choose $v_i$ to be the furthest node in $V \setminus \bigcup_{j=1}^{i-1} V_j$,
any path in $\Pc'$ containing $v_i$ can only end at a node in $\bigcup_{j=1}^{i-1}V_j$.   
Since $v_i \not\in B(v_j)$ for all $j<i$, we have 
$\sum_{P \in \mathcal P_{v_i}} x^*_P\geq 1-\sum_{j=1}^{i-1}\sum_{P\in\mathcal P_{v_j}: v_i \in P}x^*_P 
>1-\frac{i-1}{3k^*}$. 
We argue that $i<3k^*$. It suffices to show that the above partitioning process terminates
after at most $N=\floor{3k^*}$ steps.
Suppose otherwise. 
Since $\sum_{P \in \Pc_{v_i}} x^*_P>1-\frac{i-1}{3k^*}$ for all $i$, and
$\Pc_u\cap\Pc_v=\es$ for distinct nodes $u$ and $v$, we have
\[ \sum_P x^*_P\geq\sum_{i=1}^N\sum_{P\in\Pc_{v_i}}x^*_P>N-\sum_{i=1}^{N} \frac{i-1}{3k^*} 
> 3k^* - 1 - \frac{3k^*(3k^*-1)}{6k^*}\geq k^*\] 
where the last inequality holds since $k^*\geq 1$.

Now consider a part $V_i$. Each $P \in \Pc_{v_{i}}$ has length at
most $D$, so $c^\reg(P)\leq D-\dist_{v_i}$. 
Hence, $x^i=\bigl\{3k^*x^*_P\}_{P\in\Pc_{v_i}}$ is a feasible solution to the \rvrp-LP
\eqref{minklp} for the instance with root $r$, node-set $V_i$, and regret bound 
$D-\dist_{v_i}$. 
By Theorem \ref{minkapx}, we can efficiently find at most
$O(k^*\sum_{P\in\Pc_{v_i}}x^*_P)$ paths with regret at most $D-\dist_{v_i}$ covering $V_i$;
each such path has length at most $D$ since $v_i$ is the furthest node from $r$ in $V_i$.  
Doing this for every $V_i$-set, we obtain $O(k^*\sum_Px^*_P)$ feasible paths covering all
nodes. 
\end{proof}

\section{Extensions} \label{extn}

\paragraph{Additive-\kvrp.}
Recall that in additive-\kvrp, we fix the number $k\geq 1$ of rooted paths that may 
be used to cover all the nodes and seek to minimize the maximum regret of a node.
We approach \kvrp by considering a related problem, {\em min-sum (additive) \kvrp}, where
the goal is to 
minimize the sum of the regrets of the $k$ paths. 
Our techniques are versatile and yield an $O(k)$-approximation for min-sum \kvrp, which 
directly yields an $O(k^2)$-approximation for \kvrp. 
These are the {\em first} approximation guarantees for these problems, even for $k=2$.
The only previous approximation results for \kvrp were for the  
special cases of tree metrics~\cite{BockGKS11}, and when $k=1$~\cite{BlumCKLMM07}.
Partially complementing this, we prove in Section~\ref{lbounds} that a natural
LP-relaxation for \kvrp along the same lines as \eqref{minklp} and
\eqref{minreglp} has an integrality gap of $\Omega(k)$. 

As in Section \ref{sec:symmetric}, our algorithm for min-sum \kvrp is based on LP
rounding. 
Let $\mathcal C$ denote the collection of all rooted paths. 
We now consider the following LP-relaxation for the problem, where we have a variable
$x_P$ for every rooted path. We use $\OPTR$ to denote the optimal value of
\eqref{minreglp}.  

{\centering
\hspace*{-0.14in}
\begin{minipage}[t]{0.65\textwidth}
\begin{equation}
\min\quad\sum_{P\in\C} c^{\reg}(P)x_P \qquad
\text{s.t.} \qquad \sum_{P\in\C:v\in P} x_P \geq 1 \quad \forall v\in V 
\label{vcover}
\end{equation}
\end{minipage}
\begin{minipage}[t]{0.37\textwidth}
\begin{equation}
\!,\qquad \sum_{P\in\C} x_P \leq \al k, \qquad x \geq 0. \tag{P2} \label{minreglp}
\end{equation}
\end{minipage}
}

\begin{lemma} \label{reglplem}
We can use a $\gamma_{\mmep}$-approximation algorithm for unweighted \mep to compute, for
any $\e>0$, a solution $x^*$ satisfying
\eqref{vcover}, $\sum_{P\in\C}x^*_P\leq\frac{k}{1-\e}$, and 
$\sum_{P\in\C} c^{\reg}(P)x^*_P\leq\frac{\gamma_{\mmep}}{1-\e}\cdot\OPTR$, in
time $\poly\bigl(\text{input size},\frac{1}{\e}\bigr)$. 
\end{lemma}

Lemma~\ref{reglplem} is proved in Section~\ref{lpsolve}. 
Let $k^*=\sum_{P\in\C}x^*_P$ and $\nu^*=\sum_{P\in\C}c^{\reg}(P)x^*_P$. 
The rounding procedure in Section~\ref{sec:symmetric} yields a bicriteria
approximation. Choosing threshold $\dt=1-\e$ to define the cut-requirement function in
step A1 (see Section~\ref{improve}) yields $\ceil{\frac{k^*}{\dt}}=\ceil{(1+O(\e))k}$
paths with total regret at most 
$\bigl(\frac{1}{\dt}+\frac{6}{1-\dt}\bigr)\nu^*=O\bigl(\frac{1}{\e}\bigr)\OPTR$. 

To obtain a true approximation, we choose $\e$ in Lemma~\ref{reglplem} so that
$k^*\leq k+\frac{1}{3}$ 
and set the threshold $\dt$ to be $1-\frac{1}{3k+2}$. Steps A1 and A2 of
Algorithm~\ref{mainalg} then yield
a forest $F$ such that $c(F)\leq\frac{3}{1-\dt}\cdot \nu^*=3(3k+2)\nu^*$, a set $W$ of witness
nodes, and an acyclic flow $z$ such that $z^{\into}(w)\geq\dt$ for all $w\in W$. 
The flow $\hz=z/\dt$ is a flow of value $k'\leq k^*/\dt\leq k+\frac{2}{3}$. 
But instead of using this to obtain an integral flow of value at most $\ceil{k'}$, we use
the integrality property of flows in a more subtle manner. 
We may decompose $\hz$ into a convex combination
of integral flows $\tz_1,\ldots,\tz_\ell$ such that 
each $\tz_i$ is a flow of value at least $\floor{k'}$ satisfying $\tz_i^{\into}(w)\geq 1$ for all
$w\in W$. Therefore the convex combination must place a weight of at least $\frac{1}{3}$
on the $\tz_i$ flows that have value at most $k$. 
Choose the flow of value at most $k$ with
smallest $c^{\reg}$-cost, and decompose this into $k''\leq k$ rooted paths $\hP_1,\ldots,\hP_{k''}$
so that (maybe after some shortcutting) every node of $W$ lies on exactly one $\hP_i$
path. 
It follows that the total $c^{\reg}$-cost of $\hP_1,\ldots,\hP_{k''}$ is at most
$3\cdot\sum_{a\in H}c^{\reg}_a\hz_a
\leq 3\cdot\frac{3k+2}{3k+1}\cdot\sum_{a\in H}c^{\reg}_az_a
\leq 4\sum_{P\in\C}c^{\reg}(P)x^*_P$. Now we apply step A4 to obtain the final set of
paths $\tP_1,\ldots,\tP_{k''}$.

\begin{theorem} \label{thm:minmax} \label{symregapx}
The above algorithm returns at most $k$ rooted paths having total regret 
$O(k)\cdot\nu^*=O(k\cdot\gm_{\mmep})\cdot\OPTR$. Thus, we obtain an $O(k)$-approximation
algorithm for min-sum \kvrp. 
This leads to an $O(k^2)$-approximation for \kvrp. 
\end{theorem}

\begin{proof} 
The total regret of $\tP_1,\ldots,\tP_{k''}$ is at most 
$\sum_{i=1}^{k''} c^{\reg}(\hP_i)+\sum_{Z\in\comp(F)}c\bigl(h(Z)\bigr)\leq\bigl(4+6(3k+2)\bigr)\nu^*$. 
\end{proof}

\paragraph{Multiplicative-\kvrp.}
This is the version of \kvrp with multiplicative regret. 
We can obtain a constant-factor approximation for multiplicative-\kvrp as follows.
Recall that $G=(V\cup\{r\},E)$ is the underlying graph.
Let $\ioptr$ be the optimal value of the multiplicative-\kvrp problem, which we may assume
we know within a $(1+\e)$-factor. 
Given an integer ``guess'' $R$ of this optimum value, 
we consider the following feasibility problem for multiplicative-\kvrp,
which is an adaptation of the LP-formulation in~\cite{ChakrabartyS11} for the $k$-route
minimum-latency problem. Let $\Time=R\cdot\max_v\dist_v$. 
We use $t$ to index the times in $\{1,\ldots,\Time\}$. 
(Recall that all $c_{uv}$s are positive integers.)
We have variables $x_{v,t}$ for every node $v$ and time $t\in[\dist_v,\ioptr\cdot\dist_v]$
to denote that $v$ is visited at time $t$. We also have variables $z_{e,t}$ for every edge
$e=(u,v)$ and time $t$ to denote that $e$ lies on some path and both $u$ and $v$ are
visited {\em by} time $t$.
\begin{gather}
\begin{split}
\sum_t x_{v,t} & = 1, \quad 
x_{v,t}=0\ \ \text{if $t\notin[\dist_v,R\cdot\dist_v]$} \qquad \forall v; 
\qquad \qquad \sum_{e} d_ez_{e,t} \le kt \quad \forall t \\ 
\sum_{e\in \delta(S)} z_{e,t} & \ge \sum_{t'\le t} x_{v,t'} 
\qquad \forall t, S\subseteq V, v\in S; 
\qquad \qquad \qquad \qquad x, z\geq 0. 
\end{split} \tag{P4} \label{multklp}
\end{gather}

\newcommand{\multklp}[1]{\ensuremath{\text{\eqref{multklp}}_{#1}}}

The constraints encode that: (i) every node has multiplicative regret at most $R$; 
(ii) the total cost of the portion of the $k$ paths up to time $t$ does not exceed $kt$;
and (iii) if a node $v$ is visited by time $t$ then there must be a path of edges
traversed by time $t$ connecting $r$ to $v$. 
We use $\multklp{R}$ to denote the above feasibility program with regret-bound $R$.

Chakrabarty and Swamy~\cite{ChakrabartyS11} show that given a feasible solution $(x,z)$, 
to $\multklp{R}$ one can obtain $k$ rooted paths covering all nodes such that the visiting
time of each client is $O(1)\cdot\sum_t tx_{v,t}\leq O(1)\cdot R\cdot\dist_v$. Thus, we
obtain an $O(1)$-approximation provided we can solve $\multklp{R}$. 
We will not quite be able to do this, but as in~\cite{ChakrabartyS11}, we
argue that despite the pseudopolynomial size of $\multklp{R}$, if it is feasible then 
one can efficiently compute a feasible solution 
to $\multklp{(1+\e)R}$, for any $\e>0$. 

Define $\Time_i=\ceil{(1+\e)^i}$, and let $\TS:=\{\Time_0,\Time_1,\ldots,\Time_\ell\}$ where
$\ell$ is the smallest integer such that $\Time_\ell\geq R\cdot\max_v\dist_v$. 
Let $\multklp{R}^{\TS}$ denote $\multklp{R}$ when we only consider times $t\in\TS$, and
now enforce that $x_{v,t}=0$ if $t<\dist_v$ or $t$ is larger than the smallest value in
$\TS$ that exceeds $R\cdot\dist_v$.

Given a feasible solution $(x,z)$ to $\multklp{R}$, we can obtain a feasible solution
$(x',z')$ to $\multklp{R}^{\TS}$ as follows. For $t=\Time_i\in\TS$, we set
$z'_{e,t}=z_{e,\min(t,\Time)}$ for every $e$, and
$x'_{v,t}=\sum_{t'=\Time_{i-1}+1}^{\min(t,\Time)}x_{v,t'}$. 
It is easy to see that $(x',z')$ can be extended to a feasible solution $(x'',z'')$ to 
$\multklp{(1+\e)R}$ by ``padding'' it suitably: set $x''_{v,t}=x'_{v,t}$ if $t\in\TS$ and 
0 otherwise, and $z''_{e,t}=z'_{e,\Time_i}$ for $t\in[\Time_i,\Time_{i+1})$.
Finally, observe that $\multklp{R}^{\TS}$ is a polynomial-size feasibility program, so one
can efficiently solve it.

To summarize, if $\multklp{R}$ is feasible, then so is $\multklp{R}^{\TS}$ and one can
compute a feasible solution to $\multklp{R}^{\TS}$ efficiently, which can be rounded to
obtain a solution with multiplicative regret $O(1)\cdot R(1+\e)$.

The same approach yields an $O(1)$-approximation for the problem of minimizing
a weighted sum of multiplicative regrets (with nonnegative weights). The only change is
that instead of the feasibility program \eqref{multklp}, we now have an LP whose
constraints are given by \eqref{multklp} and whose objective function is to 
minimize $\sum_{v,t}\frac{w_v}{\dist_v}\cdot tx_{v,t}$, where $w_v$ is the weight
associated with $v$'s regret. The compact LP with times in $\TS$ follows analogously, and
the rounding algorithm is unchanged. 

\begin{theorem} \label{multregapx}
There is an $O(1)$-approximation algorithm for multiplicative-\kvrp. This guarantee
extends to the setting where we want to minimize a weighted sum of the multiplicative
client-regrets (with nonnegative weights).
\end{theorem}

\paragraph{Asymmetric metrics.}
We can also consider \rvrp and \kvrp in directed graphs, that is, the  
distances $\{c_{uv}\}$ now form an asymmetric metric. 
The regret of a node $v$ with respect to a directed path $P$ rooted at $r$ is defined as
before, and we seek rooted (directed) paths that cover all the nodes.
We crucially exploit that, as noted in Fact~\ref{obs:regret}, the regret
distances $\{c^{\reg}_{uv}\}$ continue to form an asymmetric metric. Thus, we
readily obtain guarantees for asymmetric \rvrp and asymmetric min-sum \kvrp by leveraging
known results for $k$-person $s$-$t$ asymmetric TSP-path (\katspp),    
which is defined as follows: given two nodes $s$, $t$ in an asymmetric
metric and an integer $k$, find $k$ $s$-$t$ paths of minimum total cost
that cover all the nodes. Friggstad et al.~\cite{FriggstadSS10} showed how to obtain
$O(k\log n)$ $s$-$t$ paths of cost at most $O(\log n)\cdot\OPT_k$, where $\OPT_k$ is the
minimum-cost \katspp solution that uses $k$ paths; this was improved by~\cite{Friggstad11}
to the following. 

\begin{theorem}[\cite{Friggstad11}]\label{thm:katspp}
\mbox{For any $b \geq 1$, we can efficiently find at most $k+\frac{k}{b}$ paths of total cost
$O(b \cdot \log n)\cdot\OPT_k$.}  
\end{theorem}

Theorem~\ref{thm:katspp} immediately yields results for 
asymmetric min-sum \kvrp---since this is simply \katspp in the regret metric!---and hence,
for asymmetric \kvrp.  

\begin{theorem} \label{asymregapx}
There is an $O(k\log n)$-approximation algorithm for asymmetric min-sum \kvrp. This
implies an $O(k^2\log n)$-approximation for asymmetric \kvrp. 
\end{theorem} 

We now focus on {\em asymmetric \rvrp}. 
We may no longer assume that $c_{uv}>0$, but we may assume that $c_{uv}+c_{vu}>0$ as 
otherwise we can again merge nodes $u$ and $v$. Consequently, at most one of $(u,v)$ or 
$(v,u)$ may lie on a shortest rooted path, and so if $R=0$, we can again efficiently
solve the problem by finding a minimum-cardinality path cover in a DAG.
Let $\iopt$ denote the optimal value of the given asymmetric \rvrp instance.
Observe that Lemma~\ref{avg2max} (as also Lemmas~\ref{lem:redcost} and \ref{dinc})
continues to hold when $c$ is asymmetric. 
Thus, we again seek to find $\al\cdot\iopt$
paths of average regret $\beta\cdot R$, for suitable values of $\al$ and $\beta$. 
We show that this can be achieved by utilizing (even) a  bicriteria approximation
algorithm for \katspp.

\begin{theorem} \label{thm:asymmetric} \label{asymkapx}
Suppose we have an algorithm for \katspp that returns at most $\al k$ $s$-$t$ paths
covering all the nodes with total cost at most $\beta\cdot\OPT_k$. 
Then, one can achieve an $O(\al+\beta)$-approximation for asymmetric \rvrp.
Thus, the results in~\cite{FriggstadSS10,Friggstad11} yield an $O(\log n)$-approximation
for asymmetric \rvrp. 
\end{theorem}

\begin{proof}
Create an auxiliary complete digraph $H=(V_H,A_H)$, where $V_H = \{r\} \cup V \cup \{t\}$.
The cost of each arc $(u,v)$ where $u, v\in \{r\} \cup V$ is its regret distance
$c^{\reg}_{uv}$; for every $v\in\{r\}\cup V$, the cost of $(v,t)$ is 0 and the cost of
$(t,v)$ is $\infty$. One can verify that these arc costs form an asymmetric metric.

We consider all values $k$ in $1,\ldots,n$ and consider the \katspp instance specified by
$H$, start node $r$, and end node $t$. When $k=\iopt$, we know that there is a solution of
cost at most $\iopt\cdot R$, so using the given algorithm
for \katspp, we obtain at most $\al\iopt$\ \! $r\leadsto t$ paths in $H$ of total cost at
most $\beta\cdot\iopt\cdot R$. So the smallest $k$ for which we obtain at most $\al k$
paths of total cost at most $\beta\cdot k\cdot R$ satisfies $k\leq\iopt$. Removing $t$
from these (at most) $\al k\leq\al\cdot\iopt$ paths yields a solution in the original
metric having total $c^{\reg}$-cost at most $\beta\cdot\iopt\cdot R$. by
Lemma~\ref{avg2max}, this can be converted to a feasible solution using
$O\bigl((\al+\beta)\cdot\iopt\bigr)$ rooted paths. 

We can obtain an $O(\log n)$-approximation to $\OPT_k$ using (at most) $k\log n$
paths~\cite{FriggstadSS10}, or $2k$ paths (taking $b = 1$ in Theorem~\ref{thm:katspp}); 
plugging this in yields an $O(\log n)$-approximation for asymmetric \rvrp. 
\end{proof}

In Section~\ref{lbounds}, we prove that an $\al$-approximation for asymmetric \rvrp yields
a $2\al$-approximation for \atsp (Theorem~\ref{asymkhard}); 
thus an $\w(\log\log n)$-factor improvement to the approximation ratio
obtained in Theorem~\ref{asymkapx} would improve the state of the art for \atsp.     

\paragraph{Non-uniform \rvrp.}
In this broad generalization of \rvrp---which captures both multiplicative-\rvrp
and \dvrp---we have {\em non-uniform integer regret bounds} $\{R_v\}_{v\in V}$ 
and we seek the fewest number of rooted paths covering all the nodes where each node $v$
has regret at most $R_v$. 
Let $R_{\max}=\max_vR_v$ and $R_{\min}=\min_{v:R_v>0}R_v$. We apply Lemma~\ref{covering}
to the sets $V_0=\{v: R_v=0\}$, and $V_i=\{v: 2^{i-1}\leq R_v<2^i\}$ 
for $i=1,\ldots,O(\log R_{\max})$. There are at most
$O\bigl(\log_2(\frac{R_{\max}}{R_{\min}})\bigr)$ non-empty $V_i$s. Let $\iopt$ be the optimal
value. We cover $V_0$ using at most $\iopt$ shortest paths, and cover every other
$V_i$-set using $O(\iopt)$ paths of regret at most $2^{i-1}$. This yields a feasible
solution using $O\bigl(\log(\frac{R_{\max}}{R_{\min}})\bigr)\cdot\iopt$ paths.

Note that applying the set-cover greedy algorithm only yields an 
$O(\log^2 n)$-approximation, since finding a minimum-density set is now a  
{\em deadline \tsp} problem 
for which we only have an $O(\log n)$-approximation~\cite{BansalBCM04}.
 
\paragraph{Capacitated variants.}
Vehicle-routing problems are often considered in capacitated settings, where we are given
a capacity bound $C$, and a path/route is considered feasible if it contains at most $C$
nodes (and is feasible for the uncapacitated problem).
Capacitated additive-\kvrp does not admit any multiplicative approximation in
polytime, since it is \npcomplete to decide if there is a solution with zero
regret~\cite{Sanita13}.  
However, when we do not fix the number of paths, a standard
reduction~\cite{NagarajanR08,BockGKS11} 
shows that an $\al$-approximation to the uncapacitated problem yields an
$(\al+1)$-approximation to the capacitated version. 
This reduction also holds in asymmetric metrics. 
Thus, we obtain approximation ratios of $31.86$ and $O(\log n)$ for capacitated \rvrp in 
symmetric and asymmetric metrics, and an 
$O\bigl(\frac{\log D}{\log\log D}\bigr)$-approximation for capacitated \dvrp.

\paragraph{Multiple depots and/or fixed destinations.} 
A natural generalization of the rooted setting is where we have a set
$S=\{r_1,\ldots,r_p\}$ of depots/sources, and a set $T=\{t_1,\ldots,t_q\}$ of
destinations/sinks, and every path must begin at an $S$-node and end at a $T$-node (and
may contain nodes of $S\cup T$ as intermediate nodes). We call such a path an $S$-$T$
path. 
We define the regret of a node $v$ with respect to an $S$-$T$ path $P$ to be
$c_P(v)-\min_{r_i\in S}c_{r_iv}$, that is, the waiting time $v$ spends in excess of the 
{\em minimum time it takes to serve $v$ from any depot}. 
We define the regret of nodes in $S\cup T$, which may lie on multiple paths, as follows. 
The regret of a source $r_i\in S$ is the minimum regret it faces along any path
containing it (which is 0 if some path originates at $r_i$). 
The regret of a sink $t_j$ is the minimum regret it faces along any path {\em ending}
at $t_j$; this effectively means that we may assume (by shortcutting) that $t_j$ is not an
intermediate node on any $S$-$T$ path.   
We obtain two variants of \rvrp: (1) in {\em $S$-$T$ \rvrp}, the goal is to find the
minimum number of $S$-$T$ paths of regret at most $R$ that cover all nodes;
(2) in {\em \multi\ \rvrp}, the goal is the same, but we have $|S|=|T|$
and require that an $S$-$T$ path starting at $r_i$ must end at $t_i$.
 
We can reduce $S$-$T$ \rvrp to \rvrp as follows. Let $\bigl(G=(V,E),\{c_{uv}\})$ be the
underlying metric. 
We create a new root node $r$ and add edges $(r,r_i)$ with $c_{rr_i}=R$ for all 
$r_i\in S$. We also create nodes $t'_1,\ldots,t'_q$, and have an edge $(t_i,t'_i)$ with  
$c_{t_it'_i}=R$ for all $i=1,\ldots,q$. Let $H$ be the resulting (non-complete) graph.
Let $c^H$ denote the shortest-path metric of $H$. 
Observe that $\dist^H_v:=c^H_{rv}=\min_{r_i\in S}c_{r_iv}+M$. 
It is easy to see that any solution to $S$-$T$ \rvrp in $G$ translates to a \rvrp solution
in $H$. Conversely, given a \rvrp solution in $H$, we take every rooted walk $P'$ in $H$
and do the following. Note that neither $r$, nor any $t'_j$ node can be intermediate nodes
of $P'$. We remove the root $r$ and possibly the end-node of $P'$ (if this is some $t'_j$) to
obtain a path $P$ in $G$ starting at some depot $r_i$. For every $v\in P'\cap V$, we have
$c^H_{P'}(v)=c_{P}(v)+M$, so the regret of $v$ does not increase. Also, if $P'$ covers
$t'_j$ then $P$ must end at $t_j$, and moreover 
$c^{\reg}_P(t_j)=\text{regret of $t'_j$ along $P'$}$. 
Finally, shortcut $P$ past the intermediate nodes in $P\cap T$ and
extend the resulting path to end at an arbitrary sink (if it does not already do so). 
The resulting collection of paths is a feasible solution to $S$-$T$ \rvrp in $G$. 

Clearly, this reduction also works in asymmetric metrics. So Theorems~\ref{minkapx}  
and~\ref{asymkapx} yield approximation ratios of $O(1)$ and $O(\log n)$ for 
$S$-$T$ \rvrp in symmetric and asymmetric metrics respectively.
We can also consider the $S$-$T$ and \multi versions of \kvrp and min-sum \kvrp, where we
seek to cover all nodes using $k$ $S$-$T$ paths, or $k$ $S$-$T$ paths such that
paths starting at $r_i\in S$ end at $t_i\in T$, so as to minimize the maximum/total regret
of a path. 
Note that $k\geq |T|$.  
The above reduction 
preserves the number of paths that are used. Hence, this reduction can also be used for
\kvrp, and we obtain the same guarantees for the $S$-$T$ \kvrp and $S$-$T$ min-sum \kvrp
in symmetric and asymmetric metrics as those listed in Theorems~\ref{symregapx}
and~\ref{asymregapx} respectively.   

\smallskip
For \multi\ \rvrp, we can leverage our techniques to achieve an $O(q)$-approximation, where
$q=|S|=|T|$. In contrast, Theorem~\ref{multihard} shows that the \multi versions of \kvrp
cannot be approximated to any multiplicative factor in polytime; the status of \multi
asymmetric \rvrp is open. 
We formulate a configuration LP with a variable for every
$r_i$-$t_i$ path of regret at most $R$, for $i=1,\ldots,q$, and approximately solve this
LP to obtain $x^*$. We assign each node $v$ to an $(r_i,t_i)$ pair satisfying 
$\sum_{\text{$r_i$-$t_i$ paths $P$}}x^*_P\geq\frac{1}{q}$, ensuring that $r_i, t_i$ are
assigned to $(r_i,t_i)$. Let $V_i$ be the nodes assigned to $(r_i,t_i)$. Shortcut each
$r_i$-$t_i$ path to contain only nodes in $V_i$, and multiply the resulting fractional 
solution by $q$. For every $i=1,\ldots,q$, this yields a collection $x^\br i$ of
fractional regret-$R$-bounded $r_i$-$t_i$ paths covering $V_i$ (to extent of 1). We now
merge all the $r_i$s to create a supernode $r$ (modifying each fractional path
accordingly), 
and round each $x^\br i$ separately using Algorithm~\ref{mainalg}. 
Thus, for each $i=1,\ldots,q$, we obtain $O(\sum_P x^\br i_P)$ $r_i$-$t_i$ paths of regret
at most $R$ covering $V_i$, and hence $O(q\sum_P x^*_P)$ paths in all.

\section{Approximation and integrality-gap lower bounds} \label{lbounds}
We now present lower bounds on the inapproximability of \rvrp and
\kvrp, and the integrality gap of the configuration LPs considered. 
We obtain both absolute inapproximability results (assuming {\em P}$\neq$\np), and
results relating the approximability of our problems to that of other 
well-known problems. 

\begin{theorem} \label{symkhard}
It is \nphard to achieve an approximation factor better than 2 for additive- and
multiplicative- \rvrp. 
\end{theorem}

\begin{proof} 
We give simple reductions from \tsp and \tsp-path.
First, consider additive-\rvrp.
Given an instance $\bigl(G=(V,E),\{c_{uv}\}\bigr)$, where $c$ is a metric, and 
a length bound $D$, we reduce the problem of determining if there is a \tsp solution of
length at most $D$ to determining if there is an additive-\rvrp solution of value 1. It
follows that it is \nphard to approximate additive-\rvrp to a factor better than 2.

We designate an arbitrary node of $G$ as the root $r$, create a new node $r'$ and add an
edge $(r,r')$ of cost $D$. The (additive) \rvrp instance is specified by the shortest-path
metric of this new graph $H$, root $r$, and regret-bound $D$. A \tsp solution of length at
most $D$ yields a \rvrp solution using one path, where we traverse the \tsp tour starting
from $r$ and then visit $r'$. Conversely, given a \rvrp solution that uses a single path
$P$, $P$ must end at $r'$, and so removing $r'$ from $P$ yields a \tsp tour in $G$ of
length at most $D$. 

\medskip
For multiplicative-\rvrp, we consider the following decision version of \tsp-path: given 
a ``sink'' $t$, we want to decide if there is a Hamiltonian path of length at most $D$
having $t$ as one of its end points. We show that this reduces to the problem of deciding
if there is a multiplicative-\rvrp solution of value 1; hence, it is \nphard to
approximate multiplicative-\rvrp within a factor better than 2.

We add two new nodes $r$ and $r'$ to $G$. We add edges $(r,v)$ of cost $D$ for all 
$v\in V$, and edges $(r,r')$ and $(t,r')$ of cost $2D$.
The multiplicative-\rvrp instance is specified by the shortest-path metric of this new
graph $H$, root $r$ and regret-bound $2$. Note that $c_{rv}=D$ for all $v\in V$, and
$c_{rr'}=2D$. 
A \tsp-path $P$ of length $D$ ending at $t$ yields a multiplicative-\rvrp solution that uses
one path, where we move from $r$ to the start node of $P$, then traverse $P$, and finally
visit $r'$. This is feasible since the visiting time of every $v\in V$ is at most $2D$,
and the visiting time of $r'$ is at most $4D\leq 2c_{rr'}$.
Conversely, suppose we have a multiplicative-\rvrp solution that uses a single path
$Q$. Then, $Q$ must end at $r'$, otherwise some node $v\in V$ has visiting time at least
$4D>2c_{rv}$. So $Q$ must move from $r$ to some $v\in V$, then cover all the nodes in $V$
ending at $t$, and finally move from $t$ to $r'$. Thus, $Q$ restricted to $V$ yields a
Hamiltonian path $P$ in $G$ ending at $t$. The visiting time of $r'$ is 
$D+c(P)+2D\leq 2\cdot 2D$, so $c(P)\leq D$.
\end{proof}

Next, we prove that the approximability of asymmetric \rvrp 
is closely related to that of \atsp. 
In particular, this connection implies that 
improving our results for asymmetric \rvrp (Theorem~\ref{asymkapx}) by an
$\w({\log\log n})$-factor would improve the state-of-the-art for \atsp.

\begin{theorem} \label{asymkhard}
Given an $\alpha$-approximation algorithm for \rvrp in asymmetric metrics, one can 
achieve a $2\alpha$-approximation for \atsp.
\end{theorem}

\begin{proof}
Suppose we have an \atsp instance with distances $c_{uv}$ whose optimal value is
$\OPT_{\matsp}$.  For a given parameter $R$, the following algorithm will return a
solution of cost at most $2\alpha \cdot R$ provided $R \geq \OPT_{\matsp}$.  
We can then use binary search to find the smallest $R$ for which the algorithm returns a
solution of cost at most $2\alpha\cdot R$, and thus obtain an \atsp solution of cost at
most $2 \alpha\cdot\OPT_{\matsp}$. 

Fix any node to be the root $r$. The algorithm first runs the $\alpha$-approximation for
asymmetric \rvrp on the \rvrp instance specified by the metric $c$ and regret bound $R$ to
find some collection of rooted paths $P_1,\ldots,P_k$. Let $v_i$ be the end node of $P_i$. 
For each $P_i$, we add the $(v_i,r)$ arc to obtain an Eulerian graph. 
The cost of the resulting Eulerian tour is $\sum_{i=1}^k (c(P_i)+c_{v_ir})$.

We claim that if $R\geq\OPT_{\matsp}$ then this cost is at most 
$2\alpha\cdot R$. To see this, note that an optimal solution to the \atsp
instance also yields a Hamiltonian path starting at $r$ of cost at most $R$. Since the
regret cost of a rooted path is at most its cost, we infer that the optimum solution to
the asymmetric \rvrp instance with regret bound $R$ uses only 1 path. Thus, we obtain that  
$k\leq\al$. 
We know that $c(P_i)\leq\dist_{v_i}+R$, and $\dist_{v_i}+c_{v_ir}\leq\OPT_{\matsp}$ for
every $i=1,\ldots,k$.  
Thus, $\sum_{i=1}^k (c(P_i)+c_{v_ir})\leq\al(R+\OPT_{\matsp})\leq 2\al R$. 
\end{proof}

Friggstad~\cite{Friggstad11} proved a hardness result for the \multi version of \katspp.
We observe that this reduction creates a \katspp instance where the metric is essentially
a regret metric. Thus,  
we obtain the following hardness results for \multi \kvrp. 

\begin{theorem} \label{multihard}
It is \npcomplete to decide if an instance of \multi \kvrp has a solution with zero
regret. Hence, no multiplicative approximation is achievable in polytime for \multi \kvrp
and \multi min-sum \kvrp. 
Moreover, \multi\ \rvrp is NP-hard even when the regret-bound is zero.
\end{theorem}

\begin{proof} 
We dovetail the reduction in~\cite{Friggstad11}. 
Given a tripartite graph $G = (U \cup V \cup W, E)$ with $|U| = |V| = |W| = n$, the
tripartite triangle cover problem is to determine if there are $n$ vertex-disjoint cliques
(which must be copies of $K_3$) in $G$. This problem is \npcomplete~\cite{GareyJ79}.%
\footnote{Strictly speaking, \cite{GareyJ79} shows that triangle cover in general graphs
is \nphard via a reduction from {\em exact cover by 3-sets}.
However, one can easily verify that if we use the \npcomplete{} {\em 3D-matching}
problem (which is a special case of exact-cover-by-3-sets) in their reduction, then the
resulting triangle-cover instances are tripartite.} 

Create four layers of nodes $V_1, V_2, V_3, V_4$ where $V_1$ and $V_4$ are disjoint copies
of $U$, $V_2$ is a copy of $V$, and $V_3$ is a copy of $W$.  For each edge $e = (u,v) \in E$,
there is exactly one index $i$ for which the copies of $u$ and $v$ are in consecutive
layers. Without loss of generality, say $u \in V_i$ and $v \in V_{i+1}$.  Then we add an
edge between the copy of $u$ in $V_i$ and $v$ in $V_{i+1}$ with cost $R > 0$.

Finally, for each $u \in U$ we create a source-sink pair that starts at the copy of $u$ in
$V_1$ and ends at the copy of $u$ in $U_4$. Denote the shortest path metric of this graph
by $H$. As in \cite{Friggstad11}, $G$ can be covered with $n$ vertex-disjoint cliques if and
only if $H$ has a solution with maximum (or total) regret 0.

The same reduction shows that \multi\ \rvrp is \nphard even with regret bound 0. Fixing the
regret-bound to zero, $H$ has a \multi\ \rvrp solution using $k$ paths iff $G$ can be covered
with $k$ vertex-disjoint cliques.
\end{proof}

\paragraph{Integrality-gap lower bounds.}
We prove that a natural configuration-style LP-relaxation for \kvrp has an
$\Omega(k)$ integrality gap. A common technique used for min-max (or bottleneck) problems
is to ``guess'' the optimal value $B$, 
which can often be used to 
devise stronger relaxations for the problem as well as strengthen the analysis, since $B$
now serves as a lower bound on the optimal value;   
examples include the algorithms of~\cite{LenstraST90,ShmoysT93,Svensson12} 
for unrelated-machine scheduling, and~\cite{AnKS10} for bottleneck \atsp.
We show that this approach does not work for \kvrp. 
Given a guess $R$ on the maximum regret, similar to \eqref{minklp}
and \eqref{minreglp}, one can consider the following feasibility problem to determine if 
there are $k$ rooted paths in $\C_R$ that cover all nodes. (Recall that $\C_R$ is the
collection of all rooted paths with regret at most $R$.) 
\begin{equation}
\sum_{P\in\C_R:v\in P} x_P \geq 1 \quad \forall v\in V, \qquad
\sum_{P\in\C_R} x_P \leq k, \qquad x \geq 0. \tag{P3} \label{minmaxlp}
\end{equation}
Let $\optrmax$ be the smallest $R$ for which \eqref{minmaxlp} is feasible, and $\ioptr$ be
the optimal value of the \kvrp instance. We describe instances
where $\ioptr\geq k\cdot\optrmax$; 
thus, the ``integrality gap'' of \eqref{minmaxlp} is at least $k$.

\begin{theorem} \label{intgap}
For any integers $h, c\geq 1$, there is a \kvrp instance with $k=c(2h-1)$ such that
$\optrmax\leq 1$ but any integer solution with maximum regret less than $2h-1$ must use at
least $k+c$ rooted  paths. Thus, (i) $c=1$ yields $\ioptr\geq k\cdot\optrmax$; 
(ii) taking $c=h$ shows that one needs $k+c=k+\frac{k}{2h-1}$ paths to achieve maximum
regret less than $(2h-1)\optrmax$.  
\end{theorem}

\begin{proof}
Our instance will consist of copies of the following ``ladder graph'' 
$L_h=(\{r\}\cup V,E)$. We have $V=\{u_1, v_1, u_2, v_2, \ldots, u_{2h-1}, v_{2h-1}\}$. 
Define $u_0=r=v_0$. $E$ consists of the edges 
$\{(u_i,u_{i+1}),(v_i,v_{i+1}): 0\leq i<2h-1\}$, which have cost $h$, along with edges 
$\{(u_i,v_i): 1 \leq i \leq 2h-1\}$, which have unit cost (see Figure~\ref{ladder}). 
Consider the shortest path metric of $L_h$. Any rooted path that covers all nodes of $L_h$
must have regret at least $2h-1$ (and this is achieved by the path 
$r, u_1, v_1, v_2, u_2,\ldots, u_{2h-1}, v_{2h-1}$).  

\begin{figure}[ht!]
\begin{center}
\includegraphics[scale=1.0]{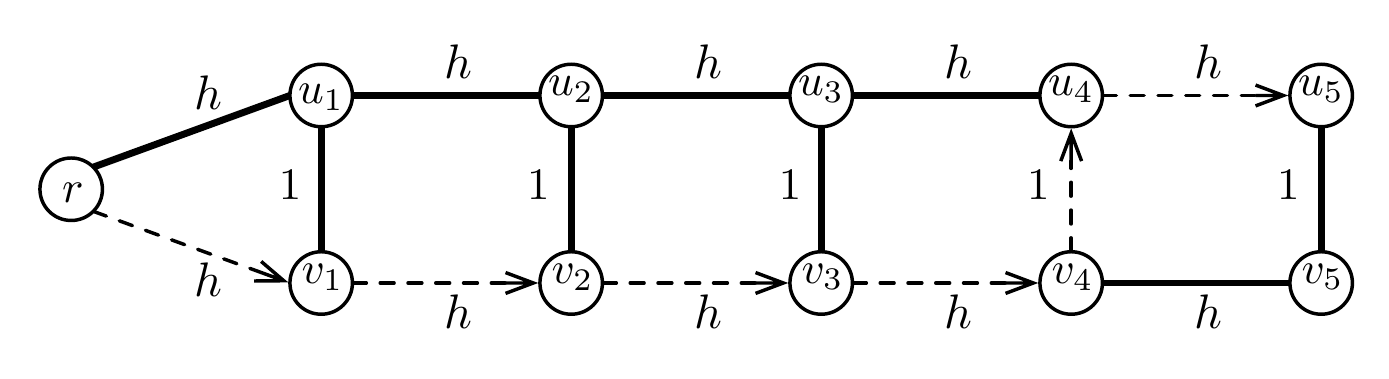}
\caption{The ladder graph $L_3$. Path $P_4$ is highlighted with dashed edges.}\label{ladder}
\end{center}
\end{figure}

Consider the paths $P_1,\ldots,P_{2h-1}$ given by
\[
P_i = \begin{cases}
r, u_1, u_2, \ldots, u_i, v_i, v_{i+1}, v_{i+2}, \ldots, v_{2h-1} & \text{if $i$ is odd} \\ 
r, v_1, v_2, \ldots, v_i, u_i, u_{i+1}, u_{i+2}, \ldots, u_{2h-1} & \text{if $i$ is even}
\end{cases}
\]
Each $P_i$ has regret exactly 1 and each node $w \neq r$ lies on precisely $h$ of
these paths. So setting $x_{P_i} = \frac{1}{h}$ for all $i=1,\ldots,2h-1$, and 
$x_P = 0$ for all other paths in $\C_1$ yields a solution that covers all the nodes in $V$
to an extent of 1 using $2-\frac{1}{h}$ paths.

The final instance consists of $ch$ copies of $L_h$ that share the root $r$ but are
otherwise disjoint. We set $k=c(2h-1)$. Taking the above fractional solution for each copy
of $L_h$, yields a feasible solution to \eqref{minmaxlp} when $R=1$. Now consider any
integer solution with maximum regret less than $2h-1$. Note that any rooted path with
regret less than $2h$ can only traverse nodes from a single ladder $L_h$. 
Also, as noted above, if a single path covers all the nodes of some copy of $L_h$, then
this path has regret at least $2h-1$. Therefore, the solution must use at least $2ch=k+c$
paths. 
\end{proof}

\section{Solving the configuration-style LPs \eqref{minklp} and \eqref{minreglp}} \label{lpsolve}  

\begin{proofof}{Lemma~\ref{klplem}}
We obtain an approximate solution to \eqref{minklp} (reproduced below), by considering the
dual problem \eqref{klpdual}, which has an exponential number of constraints.
Recall that $P$ indexes paths in $\mc C_R$.

\vspace{-5pt}
\noindent \hspace*{-3ex}
\begin{minipage}[t]{.49\textwidth}
\begin{alignat}{3}
\min & \quad & \sum_P x_P & \tag{P} \\
\text{s.t.} && \sum_{P:v\in P} x_P & \geq 1 \qquad && \forall v\in V \label{cover} \\  
&& x & \geq 0. \notag 
\end{alignat}
\end{minipage}
\quad \rule[-19ex]{1pt}{16ex}
\begin{minipage}[t]{.4\textwidth}
\begin{alignat}{3}
\max & \quad & \sum_v\pi_v & \tag{D} \label{klpdual} \\
\text{s.t.} && \sum_{v\in P}\pi_v & \leq 1 \qquad && \forall P \label{prewd} \\
&& \pi & \geq 0. \label{noneg}
\end{alignat}
\end{minipage}

\medskip \noindent
The $\pi_v$ dual variables correspond to the primal constraints \eqref{cover}. We show
that \eqref{klpdual} can be solved approximately, and hence \eqref{minklp} can be solved
approximately. Separating over the dual constraints \eqref{prewd} amounts to determining,
given rewards $\{\pi_v\}$ on the nodes, if there is a rooted path of regret at most $R$
that gathers reward more than 1. We try all possible locations $t$ for the end-node of
this path; for a given destination $t$, the above problem is an instance of \orient. 

Define $\Pc(\nu;a):=\{\pi: \eqref{prewd},\ \eqref{noneg},\ \sum_v\pi_v\geq a\nu\}$. 
Note that $\OPT$ is the largest $\nu$ such that $\Pc(\nu;1)\neq\es$. 
We use the $\gm_{\morient}$-approximation algorithm to give an {\em approximate}
separation oracle in the following sense. Given $\nu, \pi$, we either show that
$\pi/\gm_{\morient}\in\Pc(\nu;1)$, or we exhibit a hyperplane separating $\pi$ from
$\Pc(\nu;\gm_{\morient})$. Thus, for a fixed $\nu$, in polynomial time, the ellipsoid
method either certifies that $\Pc(\nu;\gm_{\morient})=\es$, or returns a point $\pi$ with
$\pi/\gm_{\morient}\in\Pc(\nu;1)$. The approximate separation oracle proceeds as
follows. We first check if $\sum_v\pi_v\geq\gm_{\morient}\nu$, \eqref{noneg} hold, and if
not, use the appropriate inequality as the separating hyperplane. 
Next, for each location $t\in V$, we run the $\gm_{\morient}$-approximation algorithm for
\orient specified by rewards $\{\pi_v\}$, $r$, $t$, and length bound $\dist_t+R$. If in
this process we ever obtain a rooted path $P$ with reward greater than 1, then $P\in\C_R$
and we return $\sum_{v\in P}\pi_v\leq 1$ as the separating hyperplane. Otherwise, for all
paths $P\in\C_R$, we have $\sum_{v\in P}\pi_v\leq\gm_{\morient}$ and so 
$\pi/\gm_{\morient}\in\Pc(\nu;1)$. 

We find the largest value $\nu^*$ (via binary search) such that the ellipsoid method run
for $\nu^*$ (with the above separation oracle) returns a solution $\pi^*$ with
$\pi^*/\gm_{\morient}\in\Pc(\nu^*;1)$; hence, $\nu^*\leq\OPT$. For any $\e>0$, running the
ellipsoid method for $\nu^*+\e$ yields a polynomial-size certificate for the emptiness of
$\Pc(\nu^*+\e;\gm_{\morient})$. This consists of the polynomially many violated
inequalities returned by the separation oracle during the execution of the ellipsoid
method and the inequality $\sum_v\pi_v\geq\gm_{\morient}(\nu^*+\e)$. By duality (or
Farkas' lemma), this means that here is a polynomial-size solution $x$ to \eqref{minklp}
whose value is at most $\gm_{\morient}(\nu^*+\e)$. Taking $\e$ to be 
$1/\exp(\text{input size})$ (so $\ln\bigl(\frac{1}{\e}\bigr)$ is polynomially bounded),
this also implies that $x$ has value at most
$\gm_{\morient}\nu^*\leq\gm_{\morient}\cdot\OPT$. 
\end{proofof}

\begin{proofof}{Lemma~\ref{reglplem}}
The proof is similar to the proof of Lemma~\ref{klplem}, but requires a more refined
argument. We again argue that an approximate solution to the dual LP \eqref{reglpdual}
can be computed efficiently. However, since the dual objective function contains negative
terms, even if our approximate separation oracle deems a point $(\pi,z)$ to be feasible, 
implying that some point in the neighborhood of $(\pi,z)$ is feasible for
\eqref{reglpdual}, we cannot necessarily claim any guarantee on the value of this dual
feasible solution relative to the value of $(\pi,z)$ (in fact its value may even be
negative!). Consequently, we will need a more refined notion of approximation for the dual
LP. This in turn will translate to a approximate solution for the primal, where the
approximation is in both the objective value and the satisfaction of the primal constraints.

\vspace{-5pt}
\noindent \hspace*{-5ex}
\begin{minipage}[t]{.49\textwidth}
\begin{alignat*}{2}
\min & \quad & \sum_{P\in\C} c^{\reg}(P)& x_P \tag{P2} \\
\text{s.t.} && \sum_{P\in\C:v\in P} x_P & \geq 1 \qquad \forall v\in V \\  
&& \sum_{P\in\C} x_P & \leq k \\
&& x & \geq 0.
\end{alignat*}
\end{minipage}
\quad \rule[-24ex]{1pt}{21ex}
\begin{minipage}[t]{.49\textwidth}
\begin{alignat}{2}
\max & \quad & \sum_v\pi_v & - kz \tag{D2} \label{reglpdual} \\
\text{s.t.} && \sum_{v\in P}\pi_v & \leq z+c^{\reg}(P) \qquad \forall P\in\C 
\label{newprewd} \\  
&& \pi, z & \geq 0. \label{newnoneg}
\end{alignat}
\end{minipage}

\medskip \noindent
Define 
$$
\Qc(\nu;a,b)\ \ :=\ \ 
\Bigl\{(\pi,z): \quad \eqref{newnoneg},
\quad \sum_{v\in P}\pi_v\leq z+\frac{c^{\reg}(P)}{b} \quad\forall P\in\C,
\quad \sum_v\pi_v-a\cdot kz\geq\nu\Bigr\}.
$$ 
So $\OPTR$ is the largest $\nu$ such that $\Qc(\nu;1,1)\neq\es$. 
For a set $S$ of nodes, let $\pi(S)$ denote $\sum_{v\in S}\pi_v$.
The separation problem for $\Qc(\nu;1,1)$ amounts to finding a rooted path that maximizes 
$\bigl(\pi(P)-c^{\reg}(P)\bigr)$.
Given a $\gm_{\mmep}$-approximation algorithm for unweighted \mep, one can obtain a path
$P$ such that 
$\pi(P)-\frac{c^{\reg}(P)}{\beta}\geq\max_{Q\in\C}\bigl(\pi(Q)-c^{\reg}(Q)\bigr)/\al$,
where $\al=(1-\e)^{-1}$ and $\beta=\gm_{\mmep}/(1-\e)$. 
Let $\pi_{\max}=\max_v\pi_v$.
We scale and round the rewards to $\pi'_v=\floor{\frac{\pi_v}{\e\pi_{\max}/n}}$.
We try every destination $t$ and run the $\gm_{\mmep}$-approximation algorithm on the
instance with rewards $\{\pi'_v\}$ (which involves making at most $\frac{n}{\e}$ copies of
a node), $r$, $t$. If an optimal solution $P^*$ ends at $t$ and achieves reward
$\Pi^*$ (note that $\Pi^*\geq\pi_{\max}$), then we obtain an $r$-$t$ path $P$ with reward at 
least $\Pi^*-\e\pi_{\max}\geq(1-\e)\Pi^*$ and $c^{\reg}(P)\leq\gm_{\mmep}\cdot c^{\reg}(P^*)$.

We use the above bicriteria approximation algorithm as follows. Given
$\nu,\pi,z$, we first check if $\sum_v\pi_v-\al kz\geq\nu$, \eqref{newnoneg} hold; if not,
we use the appropriate inequality as a separating hyperplane between $(\pi,z)$ and
$\Qc(\nu;\al,\beta)$. Next, we use the above algorithm to obtain a rooted path
$P$. If $\pi(P)-\frac{c^{\reg}(P)}{\beta}>z$, then we use
$\pi(P)\leq z+\frac{c^{\reg}(P)}{\beta}$ as a separating hyperplane between
$(\pi,z)$ and $\Qc(\nu;\al,\beta)$. Otherwise, we know that
$\pi(Q)-c^{\reg}(Q)\leq\al z$ for all $Q\in\C$, and so $(\pi,\al z)\in\Qc(\nu;1,1)$.

Thus, for a fixed $\nu$, the ellipsoid method either determines that
$\Qc(\nu;\al,\beta)=\es$ or returns a point $(\pi,\al z)\in\Qc(\nu;1,1)$. We find
the largest value $\nu^*$ for which the latter outcome is obtained. So
$\nu^*\leq\OPTR$. For any $\e>0$, running the ellipsoid method with $\nu=\nu^*+\e$ returns
a polynomial-size certificate for the emptiness of $\Qc(\nu;\al,\beta)$. 
Taking $\e=1/\exp(\text{input size})$, this shows that $\nu^*$ is an upper bound on the
maximum value of $\sum_v\pi_v-\al k z$ subject to \eqref{newnoneg} and the constraints
$\pi(P)-c^{\reg}(P)/\beta\leq z$ for every $P\in\C'$, where $\C'\sse\C$ is 
the polynomial-size collection of paths corresponding to the path-inequalities returned by
the approximate separation oracle. Taking the dual, we obtain the following LP:
$$
\min\quad\sum_{P\in\C'}\tfrac{c^{\reg}(P)}{\beta}\cdot x_P \qquad
\text{s.t.} \qquad \sum_{P\in\C':v\in P} x_P \geq 1 \quad \forall v\in V, 
\qquad \sum_{P\in\C'} x_P \leq \al k, \qquad x \geq 0.
$$
whose optimal value is at most $\nu^*$. The optimal solution $x^*$ to the above LP
satisfies the desired properties.
\end{proofof}

\end{document}